\documentclass[aps, twocolumn, nofootinbib, longbibliography]{revtex4-1}

\usepackage{amsthm}
\usepackage{amsmath}
\usepackage{amssymb}
\usepackage{bbding}
\usepackage{enumerate}
\usepackage{datetime}
\usepackage{pstricks}
\usepackage{pst-solides3d}
\usepackage{pst-poly}
\usepackage{pst-3dplot}
\usepackage{pstricks-add}
\usepackage[colorlinks=true, citecolor=red, urlcolor=blue]{hyperref}

\DeclareMathOperator{\tr}{tr}
\DeclareMathOperator{\spa}{span}
\DeclareMathOperator{\Herm}{Herm}
\DeclareMathOperator{\End}{End}

\DeclareMathOperator{\conv}{conv}

\DeclareMathOperator{\aff}{aff}
\DeclareMathOperator{\ver}{vert}

\DeclareMathOperator{\im}{im}

\newcommand{\rst}[1]{\ensuremath{{\mathbin\upharpoonright}%
\raise-.5ex\hbox{$#1$}}}

\newtheoremstyle{thm}% name
     {10pt}%      Space above
     {10pt}%      Space below
     {}%         Body font
     {}%         Indent amount (empty = no indent, \parindent = para indent)
     {\bfseries}% Thm head font
     {:}%        Punctuation after thm head
     {.5em}%     Space after thm head: " " = normal interword space;
     {}%         Thm head spec (can be left empty, meaning `normal')

\theoremstyle{thm}
\newtheorem{thm}{Theorem}

\newtheorem{lemma}[thm]{Lemma}

\newtheorem*{postulate}{Postulate}
\newtheorem*{result}{Result}
\newtheorem*{lemm}{Lemma}
\newtheorem*{defi}{Definition}

\newtheoremstyle{princ}% name
     {10pt}%      Space above
     {10pt}%      Space below
     {\itshape}%         Body font
     {}%         Indent amount (empty = no indent, \parindent = para indent)
     {\bfseries}% Thm head font
     {:}%        Punctuation after thm head
     {.5em}%     Space after thm head: " " = normal interword space;
     {}%         Thm head spec (can be left empty, meaning `normal')

\theoremstyle{princ}

\newtheoremstyle{ex}% name
     {10pt}%      Space above
     {10pt}%      Space below
     {}%         Body font
     {}%         Indent amount (empty = no indent, \parindent = para indent)
     {\bfseries}% Thm head font
     {:}%        Punctuation after thm head
     { }%     Space after thm head: " " = normal interword space;
     {}%         Thm head spec (can be left empty, meaning `normal')

\theoremstyle{ex}

\makeindex

\makeatletter
\def\@xfloat@prep{%
  \ltx@footnote@pop
  \def\@mpfn{mpfootnote}%
  \def\thempfn{\thempfootnote}%
  \c@mpfootnote\z@
  \let\H@@footnotetext\H@@mpfootnotetext
}%
\makeatother

\begin{document}

\title{If no information gain implies no disturbance, \\then any discrete physical theory is classical}
\author{Corsin Pfister}
\author{Stephanie Wehner}
\affiliation{Centre for Quantum Technologies, National University of Singapore, 3 Science Drive 2, Singapore 117543, Singapore}
\date{\today}
\begin{abstract}
	It has been suggested that nature could be discrete in the sense that the underlying state space of a physical system
	has only a finite number of pure states. For example, the Bloch ball of a single qubit could be discretized into small patches and only appear round
	to us due to experimental limitations.
	Here, we present a strong physical argument for the quantum theoretical property that every state space (even the smallest possible one, the qubit) has infinitely many 
	pure states. We propose a simple physical postulate which dictates that in fact the only possible discrete theory is classical mechanics.
	More specifically, we postulate that no information gain implies no disturbance --- or, read in the contrapositive, that 
	disturbance leads to some form of information gain. In a theory like quantum mechanics where we already know that the converse
	holds, i.e. information gain does imply disturbance, this can be understood as postulating an equivalence between disturbance and information gain. 
	What's more, we show that non-classical discrete theories are still ruled out even if we relax the postulate to hold only approximately in the
	sense that no information gain only causes a small amount of disturbance. 
	Finally, our postulate also rules out popular generalizations such as the PR-box that allows non-local correlations beyond the limits
	of quantum theory.
\end{abstract}
\maketitle

In contrast to classical theory, quantum theory has the remarkable property that the state space of every system 
has continuously many pure states. These are states which can be seen as states of maximal knowledge: They cannot be 
prepared by flipping a (possibly biased) coin to decide between two different preparation procedures to be executed, 
hiding the outcome of the coin flip. 
Even the qubit, the smallest possible system with no more than two perfectly distinguishable states, has continuously many 
such states. This non-discreteness of quantum theory contrasts sharply with classical theory, where systems with a finite 
number of perfectly distinguishable states have the same \emph{finite} number of pure states. While from a mathematical 
point of view, this quantum property is satisfactorily explained as a consequence of the mathematical framework of quantum 
theory, a physical explanation of this phenomenon is less evident.

Indeed one might conjecture that the actual state space of a physical system really was discrete with only finitely many pure states (see Fig.\ 
\ref{discretized-figure}) \cite{BHZ05, BHZ06}. The fact that experiments have not found a deviation from the continuous nature of the quantum 
state spaces could then be explained by insufficient measurement precision. A qubit, for example, could be described by 
a polytope that approximates the continuous spherical shape of the Bloch ball very well, while it actually is a discrete system. 
Quantum gravitational considerations have led some authors to the idea that indications for the discreteness of spacetime 
could in turn provide an indication for the discreteness of quantum state spaces \cite{BHZ05, BHZ06}. 
Such considerations might suggest state spaces with an extremely high number of pure states, but as long as the number 
of pure states is finite, they would differ from quantum state spaces in a fundamental way.

\begin{figure}[htb]

\begin{pspicture}[showgrid=false](-6.5,-2)(1.5,1.3)
\psset{viewpoint=50 50 20, Decran=80, linewidth=0.5\pslinewidth}
\psSolid[
	action=draw**,
	object=geode,
	ngrid = 5 1
	]
\pscircle(-5, 0){1.1}
\psellipticarc[linestyle=dashed](-5, 0)(1.1, 0.3){0}{180}
\psellipticarc(-5, 0)(1.1, 0.3){180}{0}
\uput[d](-5, -1.2){Bloch ball}
\uput[d](0, -1.2){discretized Bloch ball}
\pnode(-3, 0){left}
\pnode(-2, 0){right}
\pcline[arrowscale=2]{->}(left)(right)
\end{pspicture}

\caption{\label{discretized-figure}Illustration of discretized state spaces: 
							One might conjecture that physical state spaces are discrete in the sense that
							they only have a finite number of pure states. In such discrete state spaces, 
							the pure states are given by the corners of the state space.}
\end{figure}
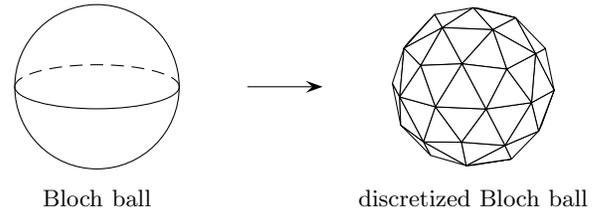

In this work, we present a strong physical counter-argument to the idea that quantum theory could be replaced by a theory with 
discrete state spaces. This argument is derived from a postulate which claims a very basic principle for measurements. 
It states that every (pure) measurement can be performed in a way such that the states with a \emph{definite} outcome
(i.e. the states with an outcome of probability one) are left invariant. We regard this principle to be a natural property of a theory that describes 
physical measurements, so we impose it as a postulate. Performing a measurement with a definite outcome does not give any 
information, while performing a measurement for which the outcome is not known in advance can be seen as a process of 
gaining information. This allows to regard our postulate as a converse 
to the well-known fact in quantum theory that information gain causes disturbance \cite{FP96}: We postulate that 
a measurement with \emph{no} information gain causes \emph{no} disturbance. We prove that a non-classical 
probabilistic theory which satisfies this postulate cannot be discrete. By a discrete system, we mean a system 
for which the state space has only finitely many pure states. In other words, we show that every theory 
which satisfies our postulate must either be classical or it must have infinitely many pure states.

\section{Technical introduction}

\subsection*{The framework}

We formulate our result in the \emph{abstract state space} framework \cite{BW09, BBLW08, BGW09, BW11}. 
This framework arises from the idea to consider the largest possible 
class of physical theories (more precisely, generalized probabilistic theories) which satisfy minimal assumptions, 
containing classical and quantum theory as special cases. This allows us to study properties of quantum theory, 
like the non-discreteness of the state space, from an outside perspective. Here we discuss these minimal 
assumptions very briefly and refer to \cite{Pfi12} for a detailed introduction to the abstract state space framework 
and its mathematical background. 

The framework, which relies on four minimal assumptions, is based on the idea that any physical theory admits the  
notions of \emph{states} and \emph{measurements}. Their interpretation is assumed to be given. The first assumption is that 
the normalized states form a convex subset $\Omega_A$ of a real vector space $A$. The underlying motivation is the idea of 
probabilistic state preparation: If $\omega, \tau \in \Omega_A$ are states which can each be prepared by a corresponding 
preparation procedure, then executing the preparation procedures with probability $p$ and $1-p$ should also lead to a state 
(described by the convex sum $p \omega + (1-p) \tau$) and should therefore be an element of $\Omega_A$ as well. 
The second assumption is that the dimension of the vector space containing the set of states is arbitrarily large but finite. 
This is a purely technical assumption intended to make the involved mathematics feasible. 
The third assumption is that the set of states $\Omega_A$ is compact. Although there might be some 
physical motivation for this assumption, we shall be satisfied with considering it as a technical assumption.

Before we discuss the fourth assumption, we make a few comments on the structure of $\Omega_A$.
The extreme points of $\Omega_A$ are the \emph{pure} states of the system, the other elements are \emph{mixed} states. 
Since $\Omega_A$ is a convex and compact subset of a finite-dimensional vector space $A$, 
every element of $\Omega_A$ is a convex combination of the extreme points 
of $\Omega_A$ \cite{Min-un, *Min1911}. 
Thus, every state is a convex combination of pure states. Since a convex combination is a sum 
with positive weights that sum up to one, a state can be seen as a probability distribution over pure states. In general, this 
probability distribution is not unique. In classical theory, however, it is (see the example below). 
In addition to the normalized states $\Omega_A$, an abstract state space $A$ also contains the 
\emph{subnormalized states} $\Omega_A^{\leq 1}$, which are given by all rescalings of the normalized states by 
factors between zero and one.

The fourth assumption states, roughly speaking, that every mathematically well-defined measurement is regarded 
as a valid measurement: A measurement is a finite set 
$\mathcal{M} = \{ f_1, \ldots, f_n \}$ of functions $f_i: A \rightarrow \mathbb{R}$ which are called \emph{effects}, 
each corresponding to an outcome of the measurement. For a state $\omega \in \Omega_A$, the value $f_i(\omega)$ 
is interpreted to be the probability that the measurement yields the outcome $i$ when the system was in the state 
$\omega$ prior to the measurement. Thus, one must have $0 \leq f_i(\omega) \leq 1$ for all $\omega \in \Omega_A$.
If the measured system was in the state $\omega$ with probability $p$ and in the state 
$\tau$ with probability $1-p$, then the probability $p f_i(\omega) + (1-p) f_i(\tau)$ of getting the outcome $i$ 
has to be identical to $f_i(p \omega + (1-p) \tau)$ since $p \omega + (1-p) \tau$ is regarded to be a state in its own 
right (in accordance with the first assumption). Skipping a few details, this means that effects are assumed to be linear. Moreover, 
the effects of a measurement have to sum up to the so-called \emph{unit effect} $\sum_{i=1}^n f_i = u_A$ 
for which $u_A(\omega) = 1$ for all $\omega \in \Omega_A$ (since the probability that \emph{any} outcome occurs 
has to be one). The fourth assumption is that every set of such linear functionals (effects) is a valid measurement. 
We denote the set of all effects on an abstract state space by $E_A$, and we denote a measurement (i.e. a set of effects 
that sum up to the unit effect) by calligraphic letters ($\mathcal{M}$ or $\mathcal{N}$ in this paper).

We would like to emphasize that the fourth assumption, which connects the geometry of the states with the geometry 
of the effects \cite{Pfi12}, is 
standard but non-trivial and of crucial technical importance for our result. A compelling physical motivation does not seem 
to be obvious, so it should be regarded as a tentative assumption on the way to a better understanding of 
quantum theory. Note that as a consequence of this assumption, a theory where the set of states is a quantum state space 
but where the measurements are restricted to a proper subset of  the positive operator valued measures (POVMs) 
is not part of the framework (c.f. quantum theory in the examples below). In quantum information science, it is always 
assumed that the full set of POVMs can be performed.

These four assumptions determine the framework of abstract state spaces. This structure is sufficient as long as 
one is only interested in measurement statistics of one-shot measurements. If one wants to describe several consecutive measurements, one has to introduce measurement-transformations. We will discuss this below, but first, we make a 
few examples.

\subsection*{Examples}

In the following, we introduce a few examples of theories which can be formulated in the abstract state space framework. 
While quantum and classical theory are theories of actual physical significance, other theories that we introduce 
play the role of toy theories which are helpful to understand the framework. Especially the square and the pentagon, 
which are instances of polygon models (see below), will serve as useful examples in the illustration of the 
proof idea of our result.

\emph{Quantum theory}: The set of states of a (finite-dimensional) quantum system is given by 
$\Omega_A = \mathcal{S}(\mathcal{H})$ for some (finite-dimensional) Hilbert space $\mathcal{H}$, where 
$\mathcal{S}(\mathcal{H})$ denotes the positive operators on $\mathcal{H}$ with unit trace (the density 
operators). These operators form a compact convex subset of $A = \Herm(\mathcal{H})$, the vector space of 
Hermitian operators on $\mathcal{H}$. Every quantum system has continuously many pure states. 
The most general description of measurement statistics in quantum theory is given by 
a POVM, which is a set $\{ F_i \}_{i=1}^n$ of positive operators which sum up to the identity operator 
$I$ on $\mathcal{H}$. They give rise to the effects $\rho \mapsto \tr(F_i \rho)$ which sum up to the unit effect 
$u_A$ given by $\rho \mapsto \tr( I \rho) = 1$ for all $\rho \in \mathcal{S}(\mathcal{H})$. In analogy to our 
comment above, we emphasize that a theory where the states form a proper subset of a quantum state space 
but where the measurements are given by not more than POVMs fails to satisfy the fourth assumption of the 
framework since a reduction of the allowed states requires an extension of the effects.

\emph{Classical theory}: The states $\Omega_A$ of a (finite) classical theory are given by a \emph{simplex}, 
that is by the convex hull of finitely many affinely independent points. (We say that points $p_1, \ldots, p_n$ in a real 
vector space are affinely independent if no point is an affine combination of the other points, i.e. if for every $p_i$, 
there are no real coefficients $\{ \alpha_k \}_{k \neq i}$ with $\sum_{k \neq i} \alpha_k = 1$ such that 
$\sum_{k \neq i} \alpha_k p_k = p_i$.) Examples of simplices are given by a line segment, a triangle, a tetrahedron, 
a pentachoron and so on. Every element of a simplex $\Omega_A$ is a \emph{unique} convex combination of 
the extreme points of $\Omega_A$ (see Fig.\ \ref{classical-property-fig}). Thus, for a simplex $\Omega_A$, the states 
are in a one-to-one correspondence with the probability distributions over the pure states, which in the case of a simplex 
are perfectly distinguishable. This allows to interpret the pure states as classical symbols. In a classical system, there is a 
generic measurement. For a given state $\omega$, the outcome probabilities for this measurement are precisely the 
coefficients in the convex sum of the pure states which yield~$\omega$.

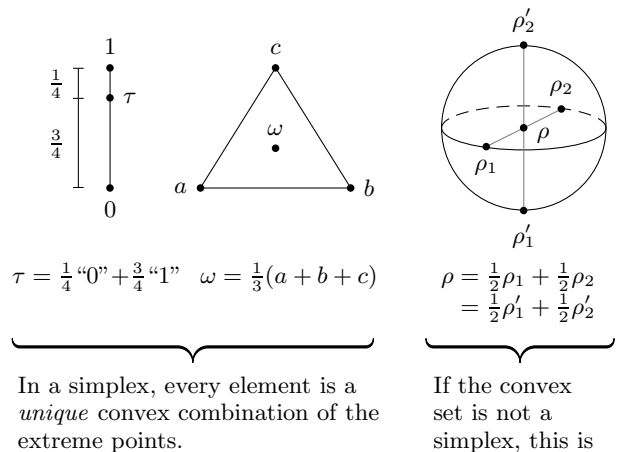
\begin{figure}[htb]

\begin{pspicture}[showgrid=false](0,-3.3)(8,3.2)

	\psset{linewidth=0.5\pslinewidth, dotsize=0.1}
	
	\newcommand{\triang}{
		\pnode(0.5,0.7){a}
		\pnode(2.5,0.7){b}
		\pnode(1.5,2.3){c}
		\pnode(1.5, 1.23){omega}
		\pspolygon(a)(b)(c)
		\psdots(a)(b)(c)(omega)
		\uput[l](a){$a$}
		\uput[r](b){$b$}
		\uput[u](c){$c$}
		\uput[u](omega){$\omega$}
	}
	\rput[bl](2,0){\triang}
	
	\newcommand{\bit}{
		\pnode(4.5, 0.7){zero}
		\pnode(4.5, 2.3){one}
		\pnode(4.5, 1.9){tau}
		\psdots(one)(zero)(tau)
		\uput[d](zero){0}
		\uput[u](one){1}
		\uput[r](tau){$\tau$}
		\pcline{-}(zero)(one)
		\pcline[offset=12pt]{|-|}(zero)(tau)
		\tlput*{$\frac{3}{4}$}
		\pcline[offset=12pt]{-|}(tau)(one)
		\tlput*{$\frac{1}{4}$}
	}
	\rput[bl](-3.2,0){\bit}
	
	\pnode(6.8, 1.5){qubit-center}
	\pscircle(qubit-center){1.1}
	\psellipticarc[linestyle=dashed](qubit-center)(1.1, 0.3){0}{180}
	\psellipticarc(qubit-center)(1.1, 0.3){180}{0}
	\pnode(6.3, 1.25){rho1}
	\pnode(7.3, 1.75){rho2}
	\pnode(6.8, 0.4){rhop1}
	\pnode(6.8, 2.6){rhop2}
	\pcline[linecolor=gray](rho1)(rho2)
	\pcline[linecolor=gray](rhop1)(rhop2)
	\psdots(qubit-center)(rho1)(rho2)(rhop1)(rhop2)
	\uput[340](qubit-center){$\rho$}
	\uput[d](rho1){$\rho_1$}
	\uput[u](rho2){$\rho_2$}
	\uput[d](rhop1){$\rho'_1$}
	\uput[u](rhop2){$\rho'_2$}
	
	\rput[tl](0,-0.3){\parbox{2.3cm}{\raggedright $\tau = \frac{1}{4}$``0''$+ \frac{3}{4}$``1''}}
	\rput[tl](2.5,-0.3){\parbox{2.5cm}{\raggedright $\omega = \frac{1}{3}(a+b+c)$}}
	\psbrace[braceWidth=0.8pt, rot=90, ref=t, nodesepB=5pt, nodesepA=5pt, braceWidthInner=0.1, braceWidthOuter=0.2]
		(0,-1.3)(4.8,-1.3){\parbox{5cm}{\raggedright 
		In a simplex, every element is a \emph{unique} convex combination of the extreme points.}}
		
	\rput[tl](5.7, -0.3){\parbox{2.3cm}{\raggedright $\rho = \frac{1}{2} \rho_1 + \frac{1}{2} \rho_2$ \\ 
		$\phantom{\rho} = \frac{1}{2} \rho'_1 + \frac{1}{2} \rho'_2$}}
	\psbrace[braceWidth=0.8pt, rot=90, ref=t, nodesepB=5pt, braceWidthInner=0.1, braceWidthOuter=0.2]
		(5.5,-1.3)(8,-1.3){\parbox{2.3cm}{\raggedright 
		If the convex set is not a simplex, this is false.}}
	
\end{pspicture}

\caption{\label{classical-property-fig} (Non-)Uniqueness of convex decompositions:
	A classical system is described by a simplex, which has the property that 
	every point is a unique convex combination of the extreme points. Thus, a state in a classical system corresponds 
	to a unique probability distribution over classical symbols.}
\end{figure}

\emph{Discrete theories}: We call $\Omega_A$ a \emph{discrete} state space if it is the convex hull of \emph{finitely} 
many (not necessarily affinely independent) points. Since $\Omega_A$ is compact, this is equivalent to saying that 
the theory has only finitely many pure states. Classical theory is an example of a discrete theory, while quantum 
theory is not.

\emph{Box world}: The \emph{generalized non-signalling theory} \cite{Bar07}, colloquially called box world, 
is a whole class of state spaces which can be formulated as abstract state spaces. It includes the well-known \emph{PR-box} 
\cite{PR94} and its local reduced state space, the so-called \emph{gbit}. All of them have only finitely 
many pure states, so they are all discrete theories according to our definition. The states $\Omega_A$ of a gbit 
form a square. Since the whole situation can be drawn in only three dimensions, the gbit provides 
an example for which we can give a picture (see Fig.\ \ref{gbit-fig}). To see the interplay of states and effects in 
such a low-dimensional example, it is useful to represent effects as vectors in the same space as the states 
\cite{JGBB11}. To evaluate an effect at some state, one simply takes the scalar product of the state and the 
vector representing the effect. This geometric representation will be useful in the illustration of the idea behind the proof 
in the Methods section.

\begin{figure}[htb]

\begin{pspicture}[showgrid=false](-2.5,0)(6,2.5)
	\psset{viewpoint=26 10 5,Decran=60}
	\psSolid[object=new, linewidth=0.5\pslinewidth,
		action=draw*,
		fcol=0 (.75 setgray),
		sommets= 
		%n=4
		0 0 0 %0
		0 1.18921 1 %w1
		-1.18921 0 1 %w2
		0 -1.18921 1 %w3
		1.18921 0 1, %w4
		faces={
		[1 2 3 4]
		[0 2 1]
		[0 3 2]
		[0 4 3]
		[0 1 4]}]%
	\rput[b](0, 2){$\Omega_A$}
	\rput[c](0, 1){$\Omega_A^{\leq 1}$}
	\rput[l](3, 2){$\Omega_A$: \parbox[t]{2cm}{\raggedright normalized states}}
	\rput[l](3, 1){$\Omega_A^{\leq1 }$: \parbox[t]{2cm}{\raggedright subnormalized states}}
	\psPoint(0,0,0){zero}
	\psdots[dotsize=0.1](zero)
	\uput[r](zero){$0$}
	\end{pspicture}
	\begin{pspicture}[showgrid=false](-2.5,0)(6,3)
	\psset{viewpoint=26 10 5,Decran=60}
	\psSolid[object=new, linewidth=0.5\pslinewidth,
		action=draw*,
		fcol=8 (.75 setgray),
		sommets= 
		%n=4
		0 0 0 %0
		0.420448 0.420448 0.5 %e1
		-0.420448 0.420448 0.5 %e2
		-0.420448 -0.420448 0.5 %e3
		0.420448 -0.420448 0.5 %e4
		0 0 1 %u
		0 1.18921 1 %w1
		-1.18921 0 1 %w2
		0 -1.18921 1 %w3
		1.18921 0 1, %w4
		faces={
		[1 4 0]
		[0 2 1]
		[0 3 2]
		[0 4 3]
		[1 2 5]
		[2 3 5]
		[3 4 5]
		[4 1 5]
		[6 7 8 9]}]%
	\psPoint(0,0,0){zero}
	\psPoint(0.594605, -0.594605, 1){state}
	\psPoint(0.420448, 0.230149, 0.5){effect}
	\pnode(4.8, 1.5){sss} %small state arrow start
	\pnode(4.2, 2.2){ssf} %small state arrow finish
	\pnode(5.4, 1.5){ses} %small effect arrow start
	\pnode(5.44, 2){sef} %small effect arrow finish
	\pnode(5, 1.8){dot}
	\psdots[dotsize=0.1](state)(effect)
	\pcline[arrowscale=1.5, linestyle=dashed]{->}(zero)(state)
	\pcline[arrowscale=1.5]{->}(zero)(effect)
	\pcline[linestyle=dashed]{->}(sss)(ssf)
	\pcline{->}(ses)(sef)
	\psdots[dotsize=0.08](dot)
	\uput[30](state){$\omega$}
	\rput[bl](effect){$f$}
	\rput[l](3,1.8){$f(\omega) =$}
	\rput[l](3, 0.9){\parbox{3cm}{\raggedright probability equals scalar product}}
	\rput[b](1,2){$\Omega_A$}
	\rput[b](-0.05,1.4){$E_A$}
	\psdots[dotsize=0.1](zero)
	\uput[r](zero){$0$}
\end{pspicture}

\caption{\label{gbit-fig}The gbit as an abstract state space: The upper part of the figure shows the 
set of normalized states $\Omega_A$ (gray), together with the subnormalized states $\Omega_A^{\leq 1}$
(white ``pyramid''), which are given by all rescalings of normalized states with factors between zero and one. 
In the lower part of the figure, the subnormalized states are omitted. Instead, the effects $E_A$ are 
shown (here they correspond to an octahedron). The reader who is familiar with the mathematics of ordered vector 
spaces may notice that the effects arise from the structure of the dual cone $A_+^*$ (more precisely, the effects 
form an order interval $[0, u_A]$ in $A^*$) \cite{Pfi12}. Here, they are represented as vectors in the same space as 
the states. To calculate a probability $f(\omega)$, one simply takes the scalar product of the vector $\omega$ and 
the vector representing $f$.}
\end{figure}
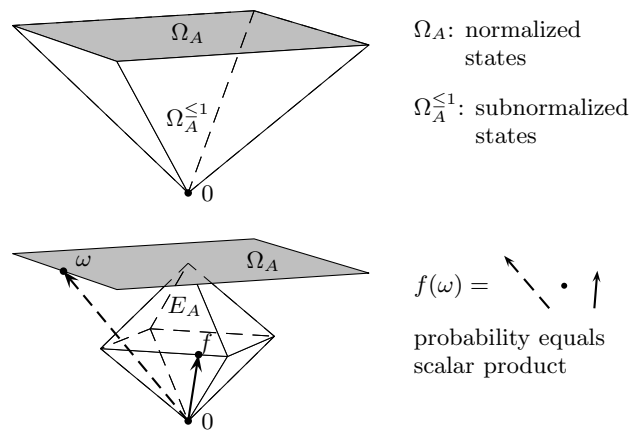

\emph{The polygon models} \cite{JGBB11}: These are abstract state spaces where $\Omega_A$ 
is a regular polygon, so they are special kinds of discrete theories. As their states and effects can 
be easily drawn in three dimensions, they also provide 
examples that we can depict. The square polygon corresponds to the gbit. In the Methods section below, 
the square and the pentagon will be the central examples in the illustration of basic idea of the proof.

\emph{Strictly convex theories}: These are theories where the set of normalized states $\Omega_A$ is a 
strictly convex set. A strictly convex set is a set where all faces are single points (if the notion of a face of 
a convex set is unknown, see \cite{Pfi12} or the Supplementary Information). In other words, a (compact) 
strictly convex set is a set such that its boundary contains no line segment, so the set is ``round'' at every 
point of the boundary. For example, the qubit, which is represented by the Bloch ball, is strictly convex, 
but every higher-dimensional quantum system is not. The latter follows from the fact that if 
$\mathcal{H}'$ is a subspace of a Hilbert space $\mathcal{H}$, then $\mathcal{S}(\mathcal{H}')$ is a face 
of $\mathcal{S}(\mathcal{H})$.

\subsection*{Post-measurement states}

In the preceding two subsections, we have discussed the core structure of abstract state spaces: states and effects. 
They only allow for the description of one-shot measurement statistics. If one wants to describe the statistics of 
\emph{several consecutive} measurements, then one has to specify what happens to the state of the system 
when a measurement is performed (otherwise, the statistics of the subsequent measurement cannot be described). 
In other words, one has to specify a rule for post-measurement states. The structure of an abstract state space, 
however, does not provide such a rule and leaves open the question of how to specify post-measurement states.

We deal with this question and consider some extra structure on abstract state spaces 
which provides a rule for post-measurement states. We describe the transition from the initial state of the system 
(prior to the measurement) to the post-measurement state by what we call a \emph{measurement-transformation}. 
Such transformations have been considered, for example, in \cite{DL70, Har01, Bar07}. We go one step further. 
Our result makes a statement about the existence of measurement-transformations in abstract state spaces 
which satisfy a certain postulate. 

As we have just mentioned above, the general idea is that a measurement-transformation 
specifies a rule for how post-measurement states are assigned. However, in a physical theory, how such a rule looks like 
depends on the particular situation which one wants to describe. To be more specific, we can think of at least three such 
situations (we will make quantum examples below), which correspond to the case where
\begin{enumerate}[(a)]
\item	the observer finds out the outcome of the measurement and describes the state of the system after the measurement
	\emph{conditioned on} that outcome.
\item	the observer describes the system after the measurement by a subnormalized state for the hypothetical case that a particular 
	outcome occurred, incorporating the probability of that outcome into the post-measurement state.
\item	the observer does \emph{not} find out the outcome of the measurement and describes the state of the system 
	after the measurement, knowing only that the measurement has been performed.
\end{enumerate}
A physical theory has to allow for a mathematical description for all of these cases. Each of the three situations can be 
described by a particular kind of map. To understand the difference between them, 
it is helpful to see how these maps look like for the particular case of quantum theory. There, if the measurement is a 
projective measurement $\mathcal{M} = \{ P_i \}_{i=1}^n$, the maps are given by \emph{L{\"u}ders projections} 
\cite{Lue50, *[{English reprint: }] Lue06} (the literature is ambiguous about which of the three maps is called a 
L{\"u}ders projection, but as they are very closely related, this usually does not lead to problems). The situations (a), (b) 
and (c) above are described by the following maps: 
\begin{enumerate}[(a)]
\item If the outcome associated with projector $P_k$ is measured, then the state is transformed as
\begin{align*}
\rho \mapsto \frac{P_k \rho P_k}{\tr(P_k \rho)} \,.
\end{align*}
\item Considering the outcome associated with projector $P_k$, the state transforms into a subnormalized state as
\begin{align*}
\rho \mapsto P_k \rho P_k \,.
\end{align*}
\item If the outcome of the measurement is unknown, the state is transformed as
\begin{align*}
\rho \mapsto \sum\limits_{i=1}^n \tr(P_i \rho) \frac{P_i \rho P_i}{\tr(P_i \rho)} = \sum\limits_{i=1}^n P_i \rho P_i \,.
\end{align*}
\end{enumerate}
Most introductory textbooks on quantum theory only discuss situation (a). Note that (a) is not a linear map. By the definition that 
we will make below, it should not be called a transformation. The maps (b) and (c) are linear. The map (b) describes what 
L{\"u}ders calls a ``measurement followed by selection'', whereas the map (c) describes what he calls a ``measurement 
followed by aggregation'' \cite{Lue50, *[{English reprint: }] Lue06}. 

The preceding discussion allows us to understand what we mean by a measurement-transformation. 
\emph{By a measurement-transformation, we mean a map of type (b)}. Note that such a map leads to 
subnormalized post-measurement states rather than normalized ones. The norm of the post-measurement state (the trace-norm 
in the quantum case) is equal to the probability that the outcome occurs (which is what we mean by ``the probability of that outcome 
is incorporated into the state''). 

Choosing maps of type (b) (rather than maps of type (a) or (c)) as the subject matter is not a relevant restriction since the three types 
of maps are so closely related that insights into one of these maps translate into insights into the other maps as well. In particular, 
from the map of type (b), one can construct the map of type (a) by rescaling the images with the inverse probability and the map of 
type (c) by summing up over all outcomes.

With the above motivation in mind, we now proceed to the task of formally defining what we mean by a measurement-transformation 
on an abstract state space. A \emph{transformation} $T$ on an abstract state space $A$ is a linear map $T: A \rightarrow A$ 
such that $T(\Omega_A) \subseteq \Omega_A^{\leq 1}$. The motivation for the linearity of transformations is similar to 
the motivation for the linearity of effects. The linearity expresses a compatibility condition for probabilistically prepared states: 
If the system is in a state $\omega$ with probability $p$ and in a state $\tau$ with probability $1-p$ before the transformation, 
then the transformed state $pT(\omega) + (1-p)T(\tau)$ has to coincide with $T(p \omega + (1-p)\tau)$ since 
$p \omega + (1-p)\tau$ is regarded as a state in its own right. (A more rigorous argument would require 
$p f(T(\omega)) + (1-p)f(T(\tau)) =  f(T(p \omega + (1-p)\tau))$ for all effects $f$, which eventually boils down to what we have 
just required.) A \emph{measurement}-transformation has to satisfy one more condition. As we have explained above, a 
measurement-transformation is associated with a particular outcome, or more precisely, with a particular effect. If $T$ is a 
measurement-transformation for an effect $f$, then we require that the norm $u_A(T(\omega))$ of the transformed state 
is equal to the probability $f(\omega)$ for measuring the outcome associated with $f$. In short, we require
\begin{align*}
	u_A \circ T = f \,.
\end{align*}
In quantum theory, where $u_A$ is given by the trace, this property is satisfied for projective measurements since 
the L\"uders projection gives $\tr(P \rho P) = \tr(P^2 \rho) = \tr(P \rho)$. 

We will only consider measurement-transformations for a special class of effects which we call \emph{pure} effects. 
We say that an effect $f \in E_A$ is pure if it is an extreme point of the (convex) set of effects $E_A$, and we say that a 
measurement $\mathcal{M} = \{f_1, \ldots, f_n\}$ is pure if every effect $f_1, \ldots, f_n$ is pure. It turns out that in the case 
of quantum theory, an effect $F \mapsto \tr(F\rho)$ of a POVM element $F$ is pure if and only if $F$ is a projector \cite{Pfi12}. 
Thus, we only consider measurement-transformations for a class of effects which, in the case of quantum theory, 
reduces to projectors. For this class, the measurement-transformations are given by L\"uders projections. 
The fact that we will restrict our considerations to pure effects is \emph{not} a restriction of the validity of our result. 
Quite the contrary, this makes our result stronger. As we will see below, our postulate claims a property of 
measurement-transformations for \emph{pure} effects rather than claiming this property for \emph{all} effects. 
This results in a \emph{weaker} postulate, so every implication derived from this postulate leads to a \emph{stronger} result. 
As we will see later, we will restrict the claim of the postulate to an even smaller subclass of effects (see the Methods section 
and the Supplementary Information for further details).

In a nutshell, a measurement-transformation for a \emph{pure} effect $f$ 
is a linear map $T: A \rightarrow A$ with $T(\Omega_A) \subseteq \Omega_A^{\leq 1}$ and $u_A \circ T = f$.

\section{Main result}

Let us first state our postulate. For a mathematically precise formulation, 
we refer to the Methods section and the Supplementary Information of this article. 

%%% XXXXXX
\begin{postulate}[No information gain implies no disturbance]
Every pure measurement can be performed in a way such that the states for which it yields a \emph{certain} outcome 
(i.e. the states with an outcome of probability one) are left invariant.
\end{postulate}

In more technical terms, the postulate states that for every pure effect $f \in E_A$, there exists an associated 
measurement-transformation $T$ with $u_A \circ T = f$ such that for all states $\omega \in \Omega_A$ 
with $f(\omega) = 1$, we have that $T(\omega) = \omega$. The existence of such a measurement-transformation $T$ is 
what is meant by saying that there exists a way to perform the measurement. 
Furthermore, note that without looking at the definition of a measurement transformation, saying that ``there exists a way to perform
the measurement'' may appear trivial by itself. After all, doing nothing and outputting the measurement outcome (associated with) $f$ 
preserves $\omega$ and yields $f$ with probability 1. This case is ruled out by the definition of a measurement-transformation. 
More precisely, note that $T$ must be such that $(u_A \circ T)(\omega') = f(\omega')$ for \emph{all} states $\omega' \in \Omega_A$. 
That is, it yields the correct probabilities for any state that we wish to measure.

It is interesting to note that the actual proof of our main result only needs an even weaker, but rather technical requirement (see the Methods section).
To see the link to information gain, note that the Shannon information content (see e.g.~\cite{mackay}) $- \log f(\omega)$ is zero for any outcome of
an experiment that occurs with certainty. As such, $f(\omega) = 1$ is equivalent to stating that no information gain occurs. The demand
that $T(\omega) = \omega$ says that the state is unchanged, i.e., no disturbance has occurred.
%%%% XXXXXXXX

Quantum theory and classical theory satisfy this postulate. In quantum theory, for example, if a system is in a
state $\rho$ such that a projective measurement $\{ P_i \}_{i=1}^n$ has some outcome $k$ with probability 
$\tr(P_k \rho) = 1$, then the transformation $\rho \mapsto P_k \rho P_k$ leaves the state invariant. 
Quantum theory even satisfies the postulate in a much stronger form in the sense that little information gain also causes only little disturbance. 
This can be seen from a special case of the \emph{gentle measurement lemma} \cite{Win99, Win98}. It states that if measuring 
an outcome associated with a projector $F$ has probability $\tr(F \rho) \geq 1 - \epsilon$, 
then measuring that outcome disturbes the state by no more than $\Vert \rho - F \rho F \Vert_1 \leq \sqrt{8 \epsilon}$. 
Setting $\epsilon = 0$, this reduces to our postulate. However, we emphasize that our 
postulate is much weaker than postulating the gentle measurement lemma. We also note that our postulate does not make
any assumptions about locality, i.e., it does not make a statement about whether verification measurements of bipartite states 
can be implemented on local quantum systems or locally disturb the state as has been considered in~\cite{popescuVaidman}.

Even though the statement of the postulate is very concise, it may appear unsatisfying since it involves the abstract concept of a state, which is something 
that one cannot observe directly. However, it can be reformulated in purely operational terms, referring only to 
directly observable objects, namely measurement statistics. Such a reformulation is possible because two states 
can be regarded as being identical if and only if they induce the same measurement statistics for every measurement (in more mathematical 
terms, a state $\omega$ is an equivalence class under the relation 
$\omega \sim \omega' \Leftrightarrow (f(\omega) = f(\omega')$ for all $f \in E_A$))~\cite{ts:entropy}. Hence, instead of making statements 
about states, one can make statements about the statistics of all potential measurements. Figure \ref{state-to-stat-fig} 
illustrates the idea of this reformulation. 

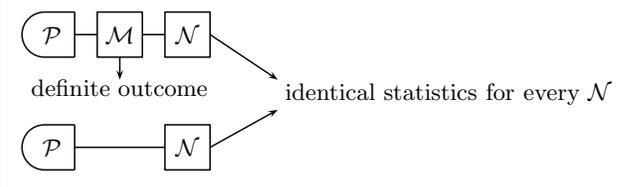
\begin{figure}[htb]

\fbox{
\begin{pspicture}[showgrid=false](0, -10)(8, 5.5)
	\psset{linewidth=0.7\pslinewidth}
	\newcommand{\prepbox}[1]{
		\pnode(0.4,0.2){boxbl}
		\pnode(0.8,0.2){boxbr}
		\pnode(0.8,0.8){boxtr}
		\pnode(0.4,0.8){boxtl}
		\pcline(boxbl)(boxbr)
		\pcline(boxbr)(boxtr)
		\pcline(boxtr)(boxtl)
		\psarc(0.4, 0.5){0.3}{90}{270}
		\rput[c](0.5,0.5){$\mathcal{#1}$}
	}
	\newcommand{\measbox}[1]{
		\pnode(0.2,0.2){boxbl}
		\pnode(0.8,0.2){boxbr}
		\pnode(0.8,0.8){boxtr}
		\pnode(0.2,0.8){boxtl}
		\pspolygon(boxbl)(boxbr)(boxtr)(boxtl)
		\rput[c](0.5,0.5){$\mathcal{#1}$}
	}
	\newcommand{\myline}[2]{
		\pnode(#1){node1}
		\pnode(#2){node2}
		\pcline(node1)(node2)
	}
	\newcommand{\myoline}[3]{
		\pnode(#1){node1}
		\pnode(#2){node2}
		\pcline(node1)(node2)
		\naput[labelsep=2pt]{#3}
	}
	\newcommand{\myodline}[3]{
		\pnode(#1){node1}
		\pnode(#2){node2}
		\pcline[linestyle=dashed](node1)(node2)
		\naput[labelsep=2pt]{#3}
	}
	\newcommand{\myncputline}[3]{
		\pnode(#1){node1}
		\pnode(#2){node2}
		\pcline(node1)(node2)
		\ncput*{#3}
	}
	\newcommand{\myuputline}[3]{
		\pnode(#1){node1}
		\pnode(#2){node2}
		\pcline(node1)(node2)
		\uput[r](node2){#3}
	}
	\newcommand{\myuputlined}[3]{
		\pnode(#1){node1}
		\pnode(#2){node2}
		\pcline{->}(node1)(node2)
		\rput[t](node2){#3}
	}
	
	\rput[tl](0, 5.5){\parbox{8cm}{Consider a preparation $\mathcal{P}$ which outputs an initial state $\omega \in \Omega_A$ and
							a measurement $\mathcal{M} = \{ f_1, \ldots, f_n \}$ such that $f_k(\omega) = 1$ 
							for some $k$. According to the postulate, the state of the system after the following 
							two experiments are identical:}}
	
	\newcommand{\exponetwo}{
		\rput[bl](0,3){\prepbox{P}}
		\myncputline{0.8, 3.5}{2, 3.5}{$\omega$}
		\rput[bl](1.8,3){\measbox{M}}
		\myuputline{2.6, 3.5}{3, 3.5}{$T_k(\omega)$}
		\myuputlined{2.3, 3.2}{2.3, 2.8}{\parbox{2.5cm}{\centering outcome $k$ \\ with certainty}}
		\rput[bl](0,1){\prepbox{P}}
		\myuputline{0.8, 1.5}{3, 1.5}{$\omega$}
		\psbrace[braceWidth=0.3pt, ref=t, nodesepB=-5pt, nodesepA=30pt, braceWidthInner=0.1, braceWidthOuter=0.2]
			(4.1,1.3)(4.1,3.8){}
		\rput[l](4.6,2.55){\parbox{3cm}{\raggedright identical states, \\ $T_k(\omega) = \omega$}}
	}
	
	\rput[tl](0,-0.5){\exponetwo}
	
	\rput[tl](0, 0.1){\parbox{8cm}{Thus, if the two experiments are followed by any measurement, say 
							$\mathcal{N} = \{ g_1, \ldots, g_l \}$, then the statistics of the 
							$\mathcal{N}$-measurement coincide:}}
							
	\newcommand{\expsagain}{
		\rput[bl](0,3){\prepbox{P}}
		\myncputline{0.8, 3.5}{2, 3.5}{$\omega$}
		\rput[bl](1.8,3){\measbox{M}}
		\myuputline{2.6, 3.5}{3, 3.5}{$T_k(\omega)$}
		\myuputlined{2.3, 3.2}{2.3, 2.8}{\parbox{2.5cm}{\centering outcome $k$ \\ with certainty}}
		\rput[bl](0,1){\prepbox{P}}
		\myuputline{0.8, 1.5}{3, 1.5}{$\omega$}
		\myline{4.1, 3.5}{4.5, 3.5}
		\rput[bl](4.3, 3){\measbox{N}}
		\myline{3.6, 1.5}{4.5, 1.5}
		\rput[bl](4.3, 1){\measbox{N}}
		\pnode(4.8, 3.2){node1}
		\pnode(4.8, 2.7){node2}
		\pcline{->}(node1)(node2)
		\rput[t](5.3, 2.6){$p_{\mathcal{N}}(j) = g_j(\omega)$}
		\psbrace[braceWidth=0.3pt, ref=t, nodesepB=-5pt, nodesepA=30pt, braceWidthInner=0.1, braceWidthOuter=0.2]
			(6.3,1.3)(6.3,3.8){}
		\rput[l](6.8,2.55){\parbox{1.5cm}{\raggedright identical statistics for every measurement~$\mathcal{N}$}}
		\pnode(4.8, 1.8){node1}
		\pnode(4.8, 2.2){node2}
		\pcline{->}(node1)(node2)
	}
	
	\rput[tl](0, -5.2){\expsagain}
	
	\rput[tl](0, -4.4){\parbox{8cm}{The $\mathcal{N}$-statistics coincide 
							\emph{for every measurement} $\mathcal{N}$. This is \emph{equivalent} to saying 
							that the states prior to the $\mathcal{N}$-measurement (i.e. $T_k(\omega)$ 
							and $\omega$) are identical. Thus, we do not need to refer to states and 
							can reformulate the postulate as: \emph{If a measurement 
							has a definite outcome, then performing this measurement does not influence 
							the statistics of any subsequent measurement}. Diagrammatically, }
	}
	
	\newcommand{\finalexp}{
		\rput[bl](0,3){\prepbox{P}}
		\myline{0.8, 3.5}{1.1, 3.5}
		\rput[bl](0.9,3){\measbox{M}}
		\myline{1.7, 3.5}{2, 3.5}
		\rput[bl](1.8,3){\measbox{N}}
		\pnode(1.4, 3.2){p1}
		\pnode(1.4, 2.9){p2}
		\pcline{->}(p1)(p2)
		\rput[t](p2){definite outcome}
		\rput[bl](0,1.5){\prepbox{P}}
		\myline{0.8, 2}{2, 2}
		\rput[bl](1.8,1.5){\measbox{N}}
		\pcline{->}(2.6,3.5)(3.5,2.9)
		\pcline{->}(2.6,2)(3.5,2.5)
		\rput[tl](3.6, 2.85){identical statistics for every $\mathcal{N}$}
	}
	
	\rput[tl](0, -11.5){\finalexp}
		
\end{pspicture}
}

\caption{	\label{state-to-stat-fig}
		A reformulation of the postulate in purely operational terms: 
		Instead of referring to initial and post-measurement states, the reformulated version 
		states that a measurement with a definite outcome does not influence the statistics of 
		any subsequent measurement, so it only refers to directly observable quantities.}
\end{figure}

In terms of the postulate, our result can now be stated as follows.

\begin{result}
An abstract state space which satisfies the postulate is either non-discrete (i.e. it has infinitely many pure states) 
or it is classical.
\end{result}

This means that if a physical system is described by an abstract state space where the set of states $\Omega_A$ 
is a polytope which is not a simplex (i.e. if it is a discrete non-classical system), then it violates our postulate. 

What is more, our result is robust in the sense that discrete non-classical theories are ruled out even if the postulate is weakened to 
an approximate version. To formulate this approximate version of the result, we assume that $A$ 
is equipped with a norm $\Vert \cdot \Vert_A$. This induces a distance function 
$\text{dist}(\omega, \omega') := \Vert \omega - \omega' \Vert_A$ on $A$. We prove that for every discrete non-classical theory, 
equipped with some norm $\Vert \cdot \Vert_A$, there is a positive number $\epsilon > 0$ such that the implication
$f(\omega) = 1 \Rightarrow \Vert T(\omega) - \omega \Vert_A \leq \epsilon$ (where $T$ is the measurement-transformation 
for $f$) cannot be satisfied for every pure effect $f \in E_A$. We prove this approximate case, which is a \emph{stronger} 
version of the result, in Section B of the Supplementary Information. 

\section{Discussion}

%%% XXXXX
Our simple postulate rules out discrete non-classical theories, while classical and quantum theory satisfy the postulate. 
Read in the contrapositive, our postulate says that disturbance implies information gain. 
Any theory that does not satisfy our postulate
thus allows for disturbance without a corresponding ability of information gain. 
Note that even in a theory which a priori only defines transformations $T$, one can define effects
as $u_A \circ T$.

We also note that our postulate rules out several alternatives to quantum theory, most notably the famous PR-box~\cite{PR94,PR1,PR2}
%%% XXXXX
that allows a violation of the CHSH inequality~\cite{chsh} far beyond the limits of quantum theory. More specifically, the PR-box achieves the algebraically maximal violation of the CHSH inequality, while still respecting the law that no information can travel faster than light. This is in spirit similar to other approaches such as information-causality~\cite{infoCausality}, 
communication complexity assumptions~\cite{wim:thesis}, the assumption of local quantum mechanics~\cite{localQM} or the uncertainty principle~\cite{js:uncertainty}.
We emphasize, however, that whereas this is a nice byproduct of our result, our real aim lies in the study of local physical systems with the goal to
identify just \emph{one} postulate that sheds light on the simple question whether the state space should be discrete or continuous.
It is very satisfying that this question can be understood by introducing just a single postulate.

One may wonder whether our postulate does in fact rule out all theories but classical and quantum mechanics. 
To answer this question, let us first be more precise about what we mean by ``a theory is (not) ruled out by the postulate''. 
We mentioned 
in the preceding section that for general abstract state spaces, measurement-transformations are not specified, 
so we cannot make statements saying that \emph{the} measurement-transformations do (not) satisfy our postulate. 
Instead, we can discuss the following well-defined question: Given an abstract state space, is it true that for every pure effect, 
there \emph{exists} a measurement-transformation which satisfies our postulate? If this is the case, then we say that the 
theory \emph{can satisfy} the postulate, or that it is \emph{not ruled out} by the postulate. If this is not true, then we say that the theory 
\emph{cannot satisfy} the postulate, or that it is \emph{ruled out} by the postulate. 

This is the precise meaning of our statement that 
``discrete non-classical theories are ruled out by the postulate''. Using this terminology, we can identify a class of theories which, in 
addition to classical and quantum theory, is not ruled out by the postulate: the strictly convex theories can satisfy our postulate. 
There are more theories which can satisfy the postulate, but we do not know a concise 
classification. For example, a state space $\Omega_A$ formed like a piece of pizza is ruled out by the postulate, while a state space 
formed like an ice cream cone is not. Figure \ref{theories-overview-fig} gives an overview.

\begin{figure}[htb]

\begin{pspicture}[showgrid=false](-4,-6.5)(4,2)
	\psset{linewidth=0.5\pslinewidth}
	\psellipse[fillstyle=hlines, hatchcolor=lightgray, hatchwidth=0.12](0,0.1)(3.9, 1.8) %big ellipse, probabilistic theories
	\rput[t](0, 1.7){\psframebox[framesep=1pt,linestyle=none, fillstyle=solid, fillcolor=white]{abstract state spaces}}
	\psellipse[fillstyle=solid, linestyle=none, fillcolor=white](-1.9, -0.4)(1.5,0.7) %ellipse filling for strictly convex theories
	\psellipse[fillstyle=solid, fillcolor=white](-2, 0.5)(1, 0.6) %ellipse for quantum theory
	\psellipse(-1.9, -0.4)(1.5,0.7) %ellipse line for strictly convex theories
	\rput[c](-2, 0.12){qubit}
	\rput[c](-1.9, -0.55){\parbox{2.5cm}{\centering strictly convex theories}}
	\rput[c](-2, 0.65){\parbox{1.2cm}{\centering quantum theory}}
	\psellipse[fillstyle=solid, fillcolor=lightgray](1.7, 0)(1.7,1) %ellipse for polytopic theories
	\psellipse[fillstyle=solid, fillcolor=white](1.7, -0.3)(1, 0.4) %ellipse for classical theories
	\rput[c](1.7, -0.3){\parbox{1.2cm}{\centering classical theory}}
	\rput[c](1.7, 0.5){\psframebox[fillstyle=none, fillcolor=white, linestyle=none]{discrete theories}}
	
	\psellipse[fillstyle=solid, fillcolor=lightgray](-3,-2.5)(0.7, 0.4)
	\rput[l](-2, -2.5){\parbox{2cm}{\raggedright ruled out by the postulate}}
	\psellipse[fillstyle=solid, fillcolor=white](1,-2.5)(0.7, 0.4)
	\rput[l](2, -2.5){\parbox{2cm}{\raggedright postulate can be satisfied}}
	\psellipse[fillstyle=hlines, hatchcolor=lightgray, hatchwidth=0.12](-3,-3.5)(0.7, 0.4)
	\rput[l](-2, -3.5){\parbox{6cm}{\raggedright some theories can satisfy the postulate, some theories cannot, for example:}}
	\rput[tl](-2, -4.8){\pswedge[fillstyle=solid, fillcolor=lightgray](0, 0){1}{0}{45}}
	\rput[tl](-0.6,-4.15){\parbox{4.5cm}{\raggedright A state space with a ``piece of pizza'' form is ruled out,}}
	\newcommand{\conefig}{
		\psellipticarc(0, 0)(0.5, 0.2){0}{180}
		\psellipticarc(0, 0)(0.5, 0.2){180}{0}
		\psline(-0.5,0)(0,0.7)
		\psline(0.5,0)(0,0.7)
	}
	\rput[tl](-0.6,-5.15){\parbox{4.5cm}{\raggedright but a state space with an ``ice cream cone'' form can satisfy the postulate.}}
	\rput[tl]{180}(-1.5, -5.4){\conefig}
\end{pspicture}

\caption{\label{theories-overview-fig} An overview over the abstract state spaces ruled out by the postulate.}
\end{figure}
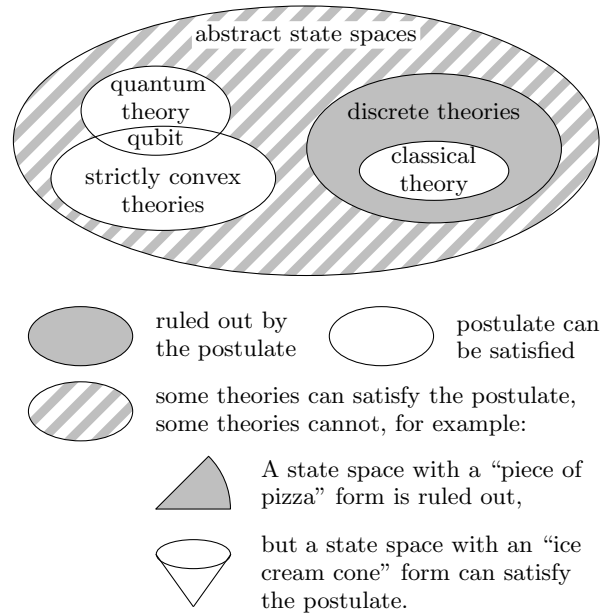

In the recent past, there have been several attempts to derive (finite-dimensional) quantum theory within a framework of 
probabilistic theories \cite{Har01, CDP11, MM11, DB11}. The idea is the following. 
One starts with a very general framework of probabilistic theories (like the abstract state space formalism). Then, one imposes 
a few physical postulates (our postulate can be seen as one such postulate). If one manages to show that 
all theories in this framework 
other than quantum theory are ruled out by these physical postulates, then this can be seen as a \emph{physical} derivation 
of quantum theory. As our postulate rules out quite a large fraction of all possible abstract state spaces already (see Fig.\ 
\ref{theories-overview-fig}), it seems promising that adding just a few more postulates might be sufficient to rule out all theories 
except for quantum theory. 

However, we do not make such an attempt and focus on one particular aspect only, introducing only one postulate. 
What makes our postulate special is that its nature is very different from the postulates 
that have been considered in this context so far. Many approaches focus on the aspect of non-locality, introducing rules for 
how physical systems are combined to form bi- or multi-partite systems. In contrast, our approach deals with local state spaces 
only, making a statement about post-measurement states. Within probabilistic theories, this aspect has gained less attention 
in the literature so far. The fact that, within the framework of abstract state spaces, we introduce just \emph{one} postulate (instead of a \emph{set} of postulates) helps us to understand its influence on one particular aspect of physical theories.

One might argue that an experimental proof of the non-discreteness of physical state spaces needs infinite measurement precision 
since the verification of the postulate that $T(\omega) = \omega$ (strict equality) requires the verification that $\omega$ and $T(\omega)$ 
give rise to the same measurement statistics (to arbitrary precision). Hence, our result is experimentally less accessible than other 
no-go theorems (e.g.\ the Bell Inequality, where it is sufficient to verify the violation of a single statistical inequality). 
There is a partial reply to this objection. As we have mentioned before, there is an approximate version 
of our result. It states that for a \emph{given} polytope $P$, there is a positive number $\epsilon_P > 0$ such that the postulate 
can be weakened to the following form (without changing the validity of the result): 
If a measurement on a state has an outcome with probability one, then performing 
the measurement does not change the state of the system by more than $\epsilon_P$ (for details, see Part B of the 
Supplementary Information). Thus, even if one weakens the postulate to allow for an $\epsilon_P$-disturbance of the state, 
it still rules out the polytope $P$. This is a \emph{stronger} form of the result. It states that in order to rule out a \emph{given} 
polytope experimentally, only finite measurement precision is needed (quantified by $\epsilon_P$). However, the allowed 
disturbance $\epsilon_P$ depends on the polytope $P$, so in order to rule out \emph{all} polytopes experimentally, infinite 
measurement precision is needed because for every measurement error, there could be a polytopic theory for the measured system 
for which the allowed disturbance $\epsilon_P$ is too small to be tested.

\section{Methods}

In this section, we sketch the idea of the proof of our main result. This will lead to geometric pictures that illustrate the 
incompatibility of non-classical discrete state spaces with our postulate (Fig.\ \ref{geom-crit-fig} and Fig.\ \ref{contra-figure}). 
For the full version of the proof and for a proof of 
the approximate version of our result, see the Supplementary Information of this article. 

Here, we aim for a geometric 
understanding of the proof. It is mainly based on a lemma which establishes geometric criteria for a set of states 
$\Omega_A$ which is compatible with our postulate. To illustrate this lemma, we provide two very basic examples 
which violate these criteria: the square and the pentagon (see Fig.\ \ref{geom-crit-fig}). For these two examples, 
it is easy to see geometrically why they cannot satisfy our postulate (as we will illustrate in Fig.\
\ref{contra-figure}). Finally, we describe roughly how we prove that every polytope $\Omega_A$ which satisfies 
the conditions of the lemma is a simplex (which is our main result).

Before we sketch the proof of the main result, it is useful to define in a bit more detail what an abstract state space is. 
For detailed definitions of the framework, see the Supplementary Information 
of this article, for a detailed motivation of the framework with detailed examples, see Chapter 3 in \cite{Pfi12}.

As illustrated in Fig.\ \ref{abs-st-sp-visualization}, an abstract state space is fully specified by a tuple $(A, A_+, u_A)$, where 
$A$ is a real finite-dimensional vector space, $A_+$ is a cone in $A$ and $u_A$ is a linear functional on $A$ 
(called the \emph{unit effect}). This linear functional is required to be strictly positive on the cone $A_+$ 
(i.e. $u_A(\omega) > 0$ for all $\omega \in A_+ \setminus \{0\}$). The tuple $(A, A_+, u_A)$ gives rise to the normalized states 
$\Omega_A$ and the subnormalized states $\Omega_A^{\leq 1}$ in the following way (c.f. Fig.\ \ref{abs-st-sp-visualization}):
\begin{align*}
&\Omega_A := \{ \omega \in A_+ \mid u_A(\omega) = 1 \} \,, \\
&\Omega_A^{\leq 1} := \{ \omega \in A_+ \mid u_A(\omega) \leq 1 \} \,.
\end{align*}

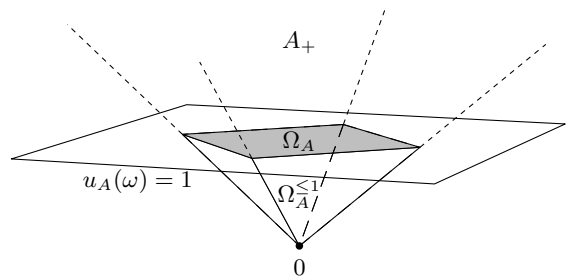
\begin{figure}[htb]
\centering

\begin{pspicture}[showgrid=false](-4,-0.5)(4,3.3)
	\psset{viewpoint=26 10 5,Decran=40}
	\psset{solidmemory}
	\psSolid[object=new, linewidth=0.5\pslinewidth,
	action=draw*,
	name=B,
	fcol=4 (.75 setgray),
	sommets= 
	%n=4
	0 0 0 %0
	0.420448 0.420448 0.5 %e1
	-0.420448 0.420448 0.5 %e2
	-0.420448 -0.420448 0.5 %e3
	0.420448 -0.420448 0.5 %e4
	0 0 1 %u
	0 1.18921 1 %w1
	-1.18921 0 1 %w2
	0 -1.18921 1 %w3
	1.18921 0 1, %w4
	faces={
	[0 7 6]
	[0 8 7]
	[0 9 8]
	[0 6 9]
	[6 7 8 9]
	}]%
	\psSolid[object=line, args=0 1.18921 1 0 2.37841 2, linestyle=dotted]
	\psSolid[object=line, args=-1.18921 0 1 -2.37841 0 2, linestyle=dotted]
	\psSolid[object=line, args=0 -1.18921 1 0 -2.37841 2, linestyle=dotted]
	\psSolid[object=line, args=1.18921 0 1 2.37841 0 2, linestyle=dotted]
	\psPoint(0,0,1.7){omegas}
	\psdots[dotsize=0](omegas)
	\uput[u](omegas){$A_+$}
	\psPoint(0,0,0){zero}
	\psdots[dotsize=0.1](zero)
	\uput[d](zero){$0$}
	\rput[b](0,0.5){$\Omega_A^{\leq 1}$}
	\psPoint(0,0,0.78){omega}
	\psdots[dotsize=0](omega)
	\uput[u](omega){$\Omega_A$}
	\psPoint(3.7,0,0.95){f}
	\psdots[dotsize=0](f)
	\uput[u](f){$u_A(\omega) = 1$}
	\psSolid[object=plan, linewidth=0.5\pslinewidth,
	definition=equation, 
	args={[0 0 1 -1]}, 
	base=-2 2 -2 2,action=draw,name=awesome]
\end{pspicture}

\caption{\label{abs-st-sp-visualization} Visualization of the state cone: The states of any normalization are given by a cone $A_+$ in the real vector space $A$. The linear functional $u_A$ gives the normalization of a state, so the intersection of $A_+$ with the plane described by $u_A(\omega) = 1$ gives the normalized states, while the subnormalized states $\Omega_A^{\leq 1}$ are those elements of $A_+$ where $u_A$ takes values between 0 and 1.}
\end{figure}

The set $E_A$ of effects on $A$ is given by the linear functionals which take values between zero and one on the states $\Omega_A$, i.e.
\begin{align*}
E_A := \{ f \in A^* \mid 0 \leq f(\omega) \leq 1 \quad \forall \omega \in \Omega_A \} \,,
\end{align*}
where $A^*$ is the dual space of $A$. A measurement is given by a finite set of effects $\mathcal{M} = \{ f_1, \ldots, f_n \} \subseteq E_A$ 
such that the effects sum up to the unit effect $u_A$, i.e. $\sum_{i=1}^n f_i = u_A$. Recall that if the system is in the state $\omega \in \Omega_A$ 
prior to the measurement described by $\mathcal{M} = \{ f_1, \ldots, f_n \}$, then the probability for outcome $k$ is given by $f_k(\omega)$.

As we have mentioned earlier, we restrict ourselves to \emph{pure} effects when we deal with post-measurement states 
(i.e. with measurement-transformations). The pure effects are the extreme points of $E_A$. A pure effect $f \in E_A$ has the property 
that the set of states $\omega$ which have probability $f(\omega) = 1$ is a \emph{face} of $\Omega_A$ \cite{Pfi12}. A face of $\Omega_A$ 
is a convex subset $F \subseteq \Omega_A$ with the property that every line segment whose endpoints are contained in $F$ must be fully 
contained in $F$, that is a face is some kind of ``extreme subset''. For a pure effect $f$, this allows us to define the \emph{certain face} $F_f$ of 
$f$ by
\begin{align*}
	F_f := \{ \omega \in \Omega_A \mid f(\omega) = 1 \} \,.
\end{align*}
Analogously, the set of states $\omega$ which have probability $f(\omega) = 0$ is a face of $\Omega_A$ as well \cite{Pfi12}. We call it 
the \emph{impossible face} of $f$ and define it by
\begin{align}
	\label{imp-eq}
	\overline F_f := \{ \omega \in \Omega_A \mid f(\omega) = 0 \} \,.
\end{align}
The notion of the certain face and the impossible face of an effect is central in our proof.

A \emph{transformation} on an abstract state space is a linear map $T: A \rightarrow A$ which is positive 
(i.e. $T(A_+) \subseteq A_+$) and does not increase the norm of the states, i.e. $u_A(T(\omega)) \leq u_A(\omega)$ for all 
$\omega \in A_+$. Equivalently, a transformation is a linear map $T: A \rightarrow A$ with $T(\Omega_A) \subseteq \Omega_A^{\leq 1}$.
Recall that we describe the state change due to a measurement by introducing \emph{measurement-transformations}. 
If a measurement yields an outcome associated to a pure effect $f \in E_A$, then the transformation of the state is described by 
$\omega \mapsto T(\omega)$, where $T$ is the measurement-transformation for $f$. As mentioned, we require that $T$ is a 
transformation which satisfies $u_A \circ T = f$.

With these definitions at hand, we can formulate our postulate as follows:

\begin{postulate}
For every pure effect $f \in E_A$, there is a transformation $T: A \rightarrow A$ such that $f = u_A \circ T$ and $T(\omega) = \omega$ 
for every $\omega \in F_f$.
\end{postulate}

Note that we only postulate the \emph{existence} of a measurement-transformation for $f$ that satisfies our postulate. 
For the actual proof, we will require an even weaker condition. We will not require the existence of such a measurement-transformation 
for every pure effect but only for pure effects for which the certain face $F_f$ is what we call a \emph{minus-face} of $\Omega_A$. This 
is a face which is exactly one dimension smaller than $\Omega_A$. This weakening of the postulate is particularly useful for the proof 
of the approximate version of our result.

To derive the result, we first prove a lemma which establishes geometric criteria which a set of states $\Omega_A$ has to satisfy 
to be compatible with our postulate. Given a pure effect $f \in E_A$, the lemma tells us geometric criteria for the 
certain face $F_f$ and the impossible face $\overline F_f$ of $f$ which are necessary for the existence of a 
measurement-transformation satisfying our postulate. It reads as follows:

\begin{lemm}
	Let $(A, A_+, u_A)$ be an abstract state space, let $f \in E_A$ be a pure effect. If there exists a transformation 
	$T: A \rightarrow A$ such that $u_A \circ T = f$ and $T(\omega) = \omega$ for every $\omega \in F_f$, then
	\begin{enumerate}[(a)]
	\item $\dim F_f + \dim \overline F_f \leq \dim \Omega_A - 1$ and
	\item if $\overline F_f$ consists of not more than one point, then $\aff(F_f \cup \overline F_f) \cap \Omega_A = \conv(F_f \cup \overline F_f)$,
	\end{enumerate}
	where $\aff( \, \cdot \, )$ and $\conv( \, \cdot \, )$ denote the affine hull and the convex hull, respectively (the reader unfamiliar 
	with these two notions is referred to the Supplementary Information of this article).
\end{lemm}

To get a geometric idea for the two conditions (a) and (b), it is useful to consider abstract state spaces which violate these 
conditions. The two simplest examples we can think of are the square and the pentagon, depicted in Fig.\ \ref{geom-crit-fig}.

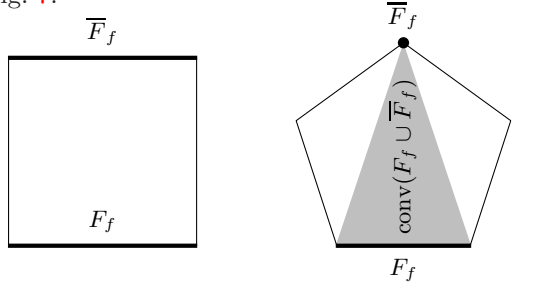
\begin{figure}[htb]
	\psset{linewidth=0.5\pslinewidth}
	\begin{pspicture}[showgrid=false](-2,-3.2)(6,1.5)
		\psline[linewidth=4\pslinewidth](-1.258,-1.25)(1.258,-1.25)
		\psline[linewidth=4\pslinewidth](-1.258,1.25)(1.258,1.25)
		\psline(-1.25,-1.25)(-1.25,1.25)
		\psline(1.25,-1.25)(1.25,1.25)
		\uput[90](0,-1.25){$F_f$}
		\uput[90](0,1.25){$\overline F_f$}
		\rput[tl](-2,-2.05){\parbox{4.2cm}{
			$\dim F_f = \dim \overline F_f = 1$, \\
			$\dim \Omega_A = 2$, so \\
			$\dim F_f + \dim \overline F_f > \dim \Omega_A - 1$
		}}
		\newcommand{\counterpenta}{
			\pspolygon[linestyle=none, fillstyle=solid, fillcolor=lightgray](-0.9,-1.25)(0,1.45)(0.9,-1.25)
			\rput(0,-0.05){\PstPentagon[PstPicture=false, unit=1.5]}
			\psline[linewidth=4\pslinewidth](-0.9,-1.25)(0.9,-1.25)
			%\psline[linestyle=dashed](-0.9,-1.25)(0,1.45)
			%\psline[linestyle=dashed](0.9,-1.25)(0,1.45)
			\psdot[dotsize=0.15](0,1.45)
			\uput[270](0,-1.25){$F_f$}
			\uput[90](0,1.45){$\overline F_f$}
			\rput[l]{90}(-0.01,-1.1){$\conv(F_f \cup \overline F_f)$}
			\rput[tl](-1.6, -2){\parbox{3.6cm}{\centering
				$\aff(F_f \cup \overline F_f) \cap \Omega_A = \Omega_A$, \\
				but $\conv(F_f \cup \overline F_f) \neq \Omega_A$
			}}
		}
		\rput[l](4,0){\counterpenta}	
	\end{pspicture}
	
	\caption{\label{geom-crit-fig} 	Violation of the conditions stated in the Lemma:
							The square and the pentagon serve as very basic examples of abstract state spaces
							which violate the conditions stated in the Lemma. The square violates condition (a), 
							while the pentagon violates (b).}
\end{figure}

To see why the conditions (a) and (b) are necessary for the existence of a transformation compatible with our postulate, 
we now examine what goes wrong in the case where one of the conditions is violated. If condition (a) is violated, 
a contradiction occurs which we call a \emph{dimension mismatch}. If (b) is violated, then we say that a \emph{shape mismatch} 
occurs. Again, the square and the pentagon serve as good examples for a geometric illustration.

\emph{Dimension mismatch}: If condition (a) is violated (i.e. $\dim F_f + \dim \overline F_f > \dim \Omega_A - 1$), then 
there is no linear map $T$ such that
\begin{align}
&u_A \circ T = f \,, \label{utf} \\
&T(\omega) = \omega \quad \text{for every} \quad \omega \in F_f \quad \text{(postulate)} \,. \label{tww}
\end{align}
In particular, there is no transformation with these two properties. To see this, there are two things to notice. 

First, Equation \eqref{utf} implies that $u_A(T(\omega)) = f(\omega) = 0$ for all $\omega \in \overline F_f$ 
(c.f. the definition \eqref{imp-eq} of $\overline F_f$). Since the zero-vector $\omega = 0$ is the only state (i.e. 
the only element of $\Omega_A^{\leq 1}$) for which $f(\omega) = 0$, it follows that the whole impossible face 
$\overline F_f$ has to be mapped to the zero-vector. By the linearity of $T$, this implies that the restriction 
$T_{\spa(\overline F_f)}$ of $T$ to $\spa(\overline F_f)$ is the zero-operator on $\spa(\overline F_f)$:
\begin{align}
	\label{square-fail-eq}
	T|_{\spa(\overline F_f)} = 0|_{\spa(\overline F_f)} \,.
\end{align}
Secondly, the postulate \eqref{tww} and the linearity of $T$ imply that the restriction 
$T|_{\spa(F_f)}$ of $T$ to $\spa(F_f)$ is the identity operator on $\spa(F_f)$:
\begin{align}
	\label{pentagon-fail-eq}
	T|_{\spa(F_f)} = I|_{\spa(F_f)} \,.
\end{align}
However, in the case where $\dim F_f + \dim \overline F_f > \dim \Omega_A - 1$, equations \eqref{square-fail-eq} and 
\eqref{pentagon-fail-eq} lead to a contradiction. In this case, the intersection $\spa(\overline F_f) \cap \spa(F_f)$ 
is a subspace which is at least one-dimensional (see Fig.\ \ref{contra-figure}). Equations \eqref{square-fail-eq} and 
\eqref{pentagon-fail-eq} imply that on this subspace, $T$ has to be the zero-operator and the identity-operator 
simultaneously, which could only be satisfied if the subspace would be $\{0\}$.

\begin{figure}[htb]
\centering

\begin{pspicture}[showgrid=false](-4,-1.5)(4,4.5)
\psset{viewpoint=26 0 5,Decran=45}
%\axesIIID[showOrigin=false](1,1,1)(3,2,2.5)
\psset{solidmemory}
\psSolid[object=new,linewidth=0.5\pslinewidth,
action=draw*,
name=B,
%fillcolor=red!50,
%fcol=8 (.5 setfillopacity Blue),
sommets= 
%n=4
0 0 0 %0
0.420448 0.420448 0.5 %e1
-0.420448 0.420448 0.5 %e2
-0.420448 -0.420448 0.5 %e3
0.420448 -0.420448 0.5 %e4
0 0 1 %u
0 1.18921 1 %w1
-1.18921 0 1 %w2
0 -1.18921 1 %w3
1.18921 0 1, %w4
fcol = 0 (.8 setgray),
faces={
%[0 9 8]
%[0 8 7]
%[0 7 6]
%[0 6 9]
[6 7 8 9]}]%
\psSolid[object=line, linewidth=2\pslinewidth,args=0 1.18921 1 1.18921 0 1]
\psSolid[object=line, linewidth=2\pslinewidth,args=0 1.18921 1.01 1.18921 0 1.01]
\psSolid[object=line, linewidth=2\pslinewidth,args=0 1.18921 1.005 1.18921 0 1.005]

\psSolid[object=line, linewidth=2\pslinewidth,args=-1.18921 0 1.005 0 -1.18921 1.005]
\psSolid[object=line, linewidth=2\pslinewidth,args=-1.18921 0 1.01 0 -1.18921 1.01]

\psPoint(0.727673, 0.528686, 1.01){g}

\rput[b](0,4.2){\parbox{8cm}{\centering Dimension mismatch:}}

\psSolid[object=line, linecolor=gray, linewidth=0.7\pslinewidth,args=-1.07029 0.713526 -0.3 0.713526 -1.07029 -0.3]
\psSolid[object=line, linecolor=gray, linewidth=0.7\pslinewidth,args=-1.07029 0.713526 -0.3 0 2.14058 1.8 2.14058 0 1.8 0.713526 -1.07029 -0.3 ]
\psSolid[object=line, linecolor=gray, linewidth=0.7\pslinewidth,args=-2.14058 0 1.8 0 -2.14058 1.8 1.07029 -0.713526 -0.3 -0.713526 1.07029 -0.3 -2.14058 0 1.8]
\psset{linewidth=1.5\pslinewidth}
\psSolid[object=line, linewidth=1.5\pslinewidth,args=-0.920449 0.920449 0 0.920449 -0.920449 0]
\psSolid[object=line, linewidth=1.5\pslinewidth,args=-0.920449 0.920449 0.005 0.920449 -0.920449 0.005]
\psSolid[object=line, linewidth=1.5\pslinewidth,args=-0.920449 0.920449 -0.005 0.920449 -0.920449 -0.005]
\psPoint(0,0,0.0){h}
\psdots[dotsize=0.12](h)
\uput[u](h){0}
\uput[u](2,2.8){$\spa(F_f)$}
\uput[u](-2,3.2){$\spa(\overline F_f)$}
\uput[u](-1.3,1.7){$\overline F_f$}
\uput[u](1.5,1.05){$F_f$}
\uput[u](0, 1.4){$\Omega_A$}
\pnode(0,-1){text}
\pnode(-0.5,-0.12){line}
\pcline[arrowscale=1.3, linewidth=0.7\pslinewidth]{->}(text)(line)
\pnode(0, -0.8){textprime}
\uput[d](textprime){$\spa(F_f) \cap \spa(\overline F_f) \neq \{0\}$}
\end{pspicture}
\begin{pspicture}[showgrid=false](-4,-0.6)(4,4.5)
\psset{viewpoint=26 10 5,Decran=70}
%\axesIIID[showOrigin=false](1,1,1)(3,2,2.5)
\psset{solidmemory}
\psSolid[object=new,linewidth=0.5\pslinewidth,
action=draw**,
name=A,
%fillcolor=red!50,
%fcol=10 (.5 setfillopacity Blue),
sommets= 
%n=5
0 0 0 %0
0.152217 0.468477 0.414214 %e1
-0.402248 0.29225 0.447214 %e2
-0.402248 -0.29225 0.447214 %e3
0.153645 -0.472871 0.447214 %e4
0.497206 0 0.447214 %e5
-0.153645 -0.472871 0.552786 %u-e1
0.402248 -0.29225 0.552786 %u-e2
0.402248 0.29225 0.552786 %u-e3
-0.153645 0.472871 0.552786 %u-e4
-0.497206 0 0.552786 %u-e5
0 0 1 %u
0.343561 1.05737 1 %w1
-0.899454 0.653491 1 %w2
-0.899454 -0.653491 1 %w3
0.343561 -1.05737 1 %w4
1.11179 0 1, %w5
fcol=5 (.55 setgray) 4 (.8 setgray),
faces={
[0 16 15]
[0 15 14]
[0 14 13]
[0 13 12]
[0 12 16]
[12 13 14 15 16]}]%
\psSolid[object=line,linestyle=dotted,args=0.343561 1.05737 1 -0.323511 1.05737 0.4]
\psSolid[object=line,linestyle=dotted,args=-0.323511 1.05737 0.4 0 0 0]
\psSolid[object=line,linestyle=dotted,args=1.11179 0 1 0.905649 -0.634423 0.4]
\psSolid[object=line,linestyle=dotted,args=0.905649 -0.634423 0.4 0 0 0]
\psSolid[object=line, linewidth=2\pslinewidth,args=0.343561 1.05737 1 1.11179 0 1]
\psSolid[object=line, linewidth=2\pslinewidth,args=0.343561 1.05737 1.01 1.11179 0 1.01]
\psSolid[object=line, linewidth=2\pslinewidth,args=0.343561 1.05737 1.005 1.11179 0 1.005]

\rput[b](0,3.7){\parbox{8cm}{\centering Shape mismatch:}}

\psPoint(0.627673, 0.628686, 1.03){f}
\psdots[dotsize=0](f)
\pnode(1,3){text2}
\pcline[arrowscale=1.3]{->}(text2)(f)
\uput[70](1,2.9){$T(F_f) = F_f$}
\psPoint(0,0,0){e3}
\psdots[dotsize=0.12](e3)
\uput[d](e3){$T(\overline F_f) = 0$}
%\pnode(-1.05,1){omegaleq}
%\psdots[dotsize=0](omegaleq)
%\uput[u](omegaleq){$\Omega_A^{\leq1}$}
%\psPoint(-0.323511, 1.05737, 0.6){omegaim}
%\psdots[dotsize=0](omegaim)
%\uput[dl](omegaim){$T(\Omega_A)$}
\psPoint(0,0,0.9){omega}
\psdots[dotsize=0](omega)
\uput[u](omega){$\Omega_A$}
\psPoint(-0.899454, -0.653491, 1){of}
\psdot[dotsize=0.12](of)
\uput[u](of){$\overline F_f$}
\psPoint(0.343561, -1.05737, 1){rho}
\psdot[dotsize=0.12](rho)
\uput[u](rho){$\rho$}
\psPoint(0.905649, -0.634423, 0.4){trho}
\psdot[dotsize=0.12](trho)
\uput[l](trho){$T(\rho)$}
\ncarc[arrowscale=1.5, arcangle=-30, nodesep=5pt]{->}{rho}{trho}
\rput[tl](-0.65,1.9){\parbox{2cm}{\centering part of $T(\Omega_A)$ \\ contained \\ in $\Omega_A^{\leq 1}$}}
\end{pspicture}

\caption{\label{contra-figure} 	Consequences of the violation of conditions (a) or (b): This figure illustrates geometrically why the square and the pentagon violate our postulate. Intuitively, all non-classical discrete state spaces exhibit either a dimension or a shape mismatch.}
\end{figure}
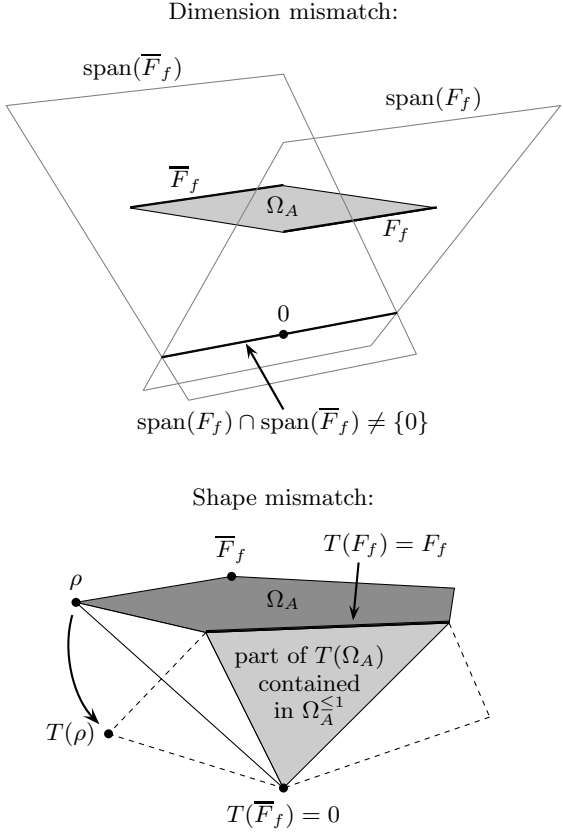

\emph{Shape mismatch}: If condition (b) is violated (i.e. if $\overline F_f$ consists of only one point and $\conv(F_f \cup \overline F_f) \neq \aff(F_f \cup \overline F_f) \cap \Omega_A$), 
then for every linear map which satisfies Equations \eqref{utf} and \eqref{tww}, there is a state $\rho$ such that 
$T(\rho) \notin \Omega_A^{\leq 1}$ (i.e. $T(\rho)$ is not a state). Therefore, such a $T$ cannot be a transformation. 
To see this geometrically, it is useful to consider the pentagon for a particular choice of the effect $f$ where the certain face $F_f$
is an edge of the pentagon (see Fig.\ \ref{contra-figure}). Equation \eqref{utf} implies that 
the impossible face $\overline F_f$ is mapped to the zero-vector, while Equation \eqref{tww} means that the certain face $F_f$ 
is left invariant. In the case of the pentagon illustrated in Fig.\ \ref{contra-figure}, there is precisely one linear map $T$ with 
these two properties. It maps the normalized states $\Omega_A$ 
(dark gray surface in the figure) to a set in the vector space (dashed lines) which is not contained in $\Omega_A^{\leq 1}$ 
(the truncated cone between $0$ and $\Omega_A$). In particular, there is a $\rho$ such that $T(\rho) \notin \Omega_A$.
If one compares Fig.\ \ref{contra-figure} with Fig.\ \ref{geom-crit-fig}, then one can see that the part of $\Omega_A$ 
which is mapped to a subset of $\Omega_A^{\leq 1}$ (light gray face in Fig.\ \ref{contra-figure}) is precisely given by 
$\conv(F_f \cup \overline F_f)$ (gray part in Fig.\ \ref{geom-crit-fig}). However, the part of $\Omega_A$ which is mapped outside of 
$\Omega_A$ is given by $(\aff(F_f \cup \overline F_f) \cap \Omega_A) \setminus \conv(F_f \cup \overline F_f)$ 
(the white part in Fig.\ \ref{geom-crit-fig}). This observation generalizes to statement (b) of the Lemma: If (a) is satisfied and $\overline F_f$ consists of only one point, then
$T(\Omega_A)$ is contained in $\Omega_A^{\leq 1}$ only if 
$\aff(F_f \cup \overline F_f) \cap \Omega_A = \conv(F_f \cup \overline F_f)$.

These two examples illustrate all that can go wrong for discrete theories. We show that for every discrete theory (i.e. for every 
theory where $\Omega_A$ is a polytope), either condition (a) or (b) is violated (so either a dimension mismatch or a shape 
mismatch occurs), except for the case where $\Omega_A$ is a simplex (i.e. for classical theories). To show this, we proceed 
as follows. 

We consider an abstract state space $(A, A_+, u_A)$ where $\Omega_A$ is a polytope. We assume that for 
every pure effect $f \in E_A$ for which the certain face $F_f$ is a minus-face of $\Omega_A$, there is a measurement-transformation 
satisfying the postulate \eqref{tww}. In a first step, we show (using the Lemma) that every polytope $\Omega_A$ which is compatible 
with our postulate has a property that we call being \emph{uniformly pyramidal}. 
This means that for every minus-face $F$ of $\Omega_A$, it holds that there is a point $a_F \in \Omega_A$ 
such that $\Omega_A = \conv(F \cup \{a_F\})$ (see the Supplementary Information for more intuition and figures). In a second step, 
we show that every uniformly pyramidal polytope $\Omega_A$ is a simplex. This shows that every discrete theory satisfying our 
postulate has to be classical.

\section*{Acknowledgements}

We thank Christian Gogolin, Paolo Perinotti, Marco Tomamichel and Markus Baden for insightful discussions and 
Matthew Pusey for interesting comments on a preliminary version of this work \cite{Pfi12}.  
We thank Dieter Kadelka for pointing out that the previous formulation of Lemma 2 (b) was incorrect (see the remark before Lemma 2 in Appendix A).
This research was supported by the Ministry of Education and the National Research Foundation, Singapore.

\bibliographystyle{arxiv}
\bibliography{discreteclassical}

\clearpage
\begin{widetext}
\section*{Supplementary Information}
\begin{appendix}

\section{Formal proof of the main theorem}
\label{main-result-appendix}

In this appendix, we prove the main result: A polytopic theory for which our postulate holds is a classical theory. A preliminary version of this proof is available at \cite{Pfi12}. The proof provided here is much more concise than the proof in \cite{Pfi12}. It makes no use of strong theorems but is entirely proved on quite an elementary level. Fewer notions are introduced, so we only make definitions that are necessary for clarification or that simplify the proof. All in all, the proof and its preparation presented here are much shorter.

We want to emphasize again that our result depends on the assumption that every mathematically well-defined measurement is allowed by the theory. As discussed in the main article, this is a standard assumption, but it lacks a clear physical motivation. One should thus be aware of the fact that our result could not be derived without this assumption.

The focus of the proof in this appendix is on technical precision. To organize this appendix in a compact way, we proceed as follows: In Section \ref{def-and-not}, we list all the definitions that are necessary to understand the proof. This list is given for technical clarification only and cannot be regarded as an introduction to the subject. For a detailed introduction to all of the notions and concepts mentioned below, we refer to \cite{Pfi12}. In Section \ref{fact-section}, we give a list of facts that we state without giving a proof here. These facts are either standard mathematical results, easy to verify or we have proved them in \cite{Pfi12}. Whenever the latter is the case, we refer to the corresponding proposition. The referenced propositions have quite elementary proofs which are not very important for the understanding of the proof of the main result. Section \ref{tech-lemma-section} is dedicated to the preparation of the main proof. We give a rough outline of what we will show and derive a few technical lemmas. Finally, we prove the main result in Section \ref{main-result-section}.

\subsection{Definitions and Notation}
\label{def-and-not}

\begin{enumerate}[(Def. 1)]

\item A subset $C$ of a real vector space $V$ is a \emph{convex subset} or \emph{convex set} if $x, y \in C$ implies $\lambda x + (1-\lambda) y \in C$ for all $0 \leq \lambda \leq 1$.

\item A nonempty convex subset $F$ of a convex set $C$ is a \emph{face} if $x, y \in C$, $0 < \lambda < 1$ and $\lambda x + (1-\lambda) y \in F$ imply $x, y \in F$. \label{face-def}

\item An element $e \in C$ of a convex set is an \emph{extreme point} of $C$ if the singleton $\{e\}$ is a face of $F$, i.e. if $x, y \in C$, $0 < \lambda < 1$ and $\lambda x + (1-\lambda) = e$ imply $x = y = e$. The set of extreme points of $C$ is denoted by $\ver(C)$. \label{ext-point-def}

\xdef\letzterwert{\the\value{enumi}}
\end{enumerate}
In the following, $V$ denotes a finite-dimensional real vector space, $S \subseteq V$ denotes any subset of $V$.
\begin{enumerate}[(Def. 1)]
\setcounter{enumi}{\letzterwert}

\item $\conv(S)$ denotes the convex hull of $S$ (see Fig. \ref{conv-aff-visualization}), given by
\begin{align*}
\conv(S) := \left\{ \sum\limits_{i=1}^n \alpha_i v_i \ \middle\vert \ n \in \{0, 1, 2, \ldots \}, v_i \in S, \alpha_i \in [0,1], \sum\limits_{i=1}^n \alpha_i = 1 \right\} \,.
\end{align*}

\item $\aff(S)$ denotes the affine hull of $S$ (see Fig. \ref{conv-aff-visualization}), given by
\begin{align*}
\aff(S) := \left\{ \sum\limits_{i=1}^n \alpha_i v_i \ \middle\vert \ n \in \{0, 1, 2, \ldots \}, v_i \in S, \alpha_i \in \mathbb{R}, \sum\limits_{i=1}^n \alpha_i = 1 \right\} \,.
\end{align*}

\xdef\letzterwertprime17{\the\value{enumi}}
\end{enumerate}

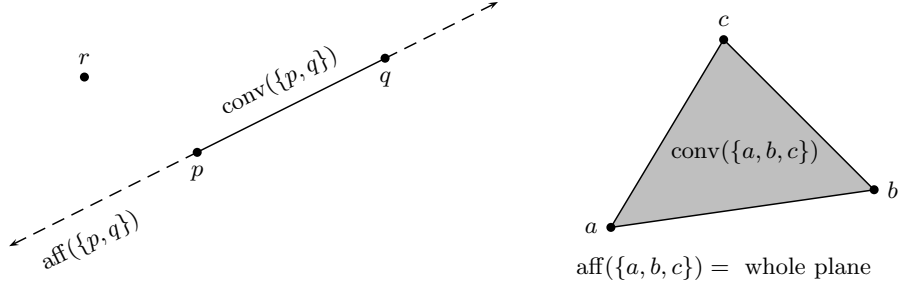
\begin{figure}
\centering

\begin{pspicture}[showgrid=false](-6,-2)(6,2)
\psset{linewidth=0.7\pslinewidth}
\pnode(-3.5,0){p}
\pnode(-1,1.25){q}
\pnode(-5,1){r}
\psdots[dotsize=0.12](p)(q)(r)
\pcline(p)(q)
\aput{:U}{$\conv(\{p, q\})$}
\pnode(-6,-1.25){pe}
\pnode(0.5,2){qe}
\pcline[linestyle=dashed]{<-}(pe)(p)
\pcline[linestyle=none, offset=12pt](pe)(p)
\nlput[offset=-12pt,nrot=:U](pe)(p){1cm}{$\aff(\{p, q\})$}
\pcline[linestyle=dashed]{->}(q)(qe)
\uput[d](p){$p$}
\uput[d](q){$q$}
\uput[u](r){$r$}

\pnode(2,-1){a}
\pnode(5.5,-0.5){b}
\pnode(3.5,1.5){c}
\pspolygon[fillstyle=solid,fillcolor=lightgray](a)(b)(c)
\rput[l](2.8,0){$\conv(\{a, b, c\})$}
\uput[d](3.5,-1.2){$\aff(\{a, b, c\}) = \text{ whole plane}$}
\psdots[dotsize=0.12](a)(b)(c)
\uput[l](a){$a$}
\uput[r](b){$b$}
\uput[u](c){$c$}
\end{pspicture}

\label{conv-aff-visualization}
\caption{Visualization of convex and affine hulls: The convex hull of two points $p$ and $q$ is given by the line segment connecting the two points, while the affine hull of two points is the whole line through the two points. The point $r$ is affinely independent of $p$ and $q$ since $r \notin \aff(\{p, q\})$. The convex hull of three points (which do not lie on a line) is given by the triangle the corners of which coincide with the three points, while the affine hull is given by the whole plane containing the three points.}
\end{figure}

\begin{enumerate}[(Def. 1)]
\setcounter{enumi}{\letzterwertprime17}

\item $\spa(S)$ denotes the linear hull (or the linear span) of $S$.

\item A point $p \in V$ is \emph{affinely independent of $S$} if $p \notin \aff(S)$. Points $p_1, \ldots, p_n \in V$ are \emph{affinely independent} if $p_i \notin \aff( \{ p_1, \ldots, p_n \} \setminus \{ p_i \})$ for $i = 1, \ldots, n$. \label{aff-indep}

\item $\dim S := d$, where $d+1$ is the maximal number of affinely independent points in $S$. (In the case where $S$ is a vector subspace of $V$, this dimension is identical to the vector space dimension of $S$.)  We define the dimension of the empty set to be $-1$. \label{set-dim-def}

\xdef\letzterwertprime2{\the\value{enumi}}
\end{enumerate}
The notion of the dimension of a subset of a vector space (Def. \ref{set-dim-def}) allows us to define the following special type of a face of a convex set (recall (Def. \ref{face-def})).
\begin{enumerate}[(Def. 1)]
\setcounter{enumi}{\letzterwertprime2}

\item We call a face $F$ of a convex set $C$ a \emph{minus-face} of $C$ if $\dim F = \dim C - 1$.\footnote{In the context of polytopes, such a face is sometimes called a \emph{facet} of $C$. However, besides the fact that the the terms ``face'' and ``facet'' are easily mixed up, the use of the notion of a facet in the literature is inconsistent, so we prefer to introduce a new name to avoid confusion.} \label{minus-face}

\xdef\letzterwertprime{\the\value{enumi}}
\end{enumerate}
The following types of convex sets are very central in our analysis.
\begin{enumerate}[(Def. 1)]
\setcounter{enumi}{\letzterwertprime}

\item A subset $P$ of a finite-dimensional real vector space is a \emph{polytope} if $P$ is the convex hull of finitely many points.

\item Since $\ver(P)$ is a finite set for any polytope $P$ (which is readily verified), we can define the number of extreme points of a polytope $P$ by $n_e(P) := \vert \ver(P) \vert$. \label{number-ext-points}

\item A subset $S$ of a finite-dimensional real vector space $V$ is a \emph{simplex} if $S$ is the convex hull of finitely many affinely independent points. More precisely, $S$ is a $d$-simplex if it is the complex hull of $d+1$ affinely independent points. Obviously, a simplex is a polytope. \label{simplex-def}

\xdef\letschtwert{\the\value{enumi}}
\end{enumerate}
Before we can give a precise definition of an abstract state space, we introduce some notation and some notions related to vector spaces with a cone.
\begin{enumerate}[(Def. 1)]
\setcounter{enumi}{\letschtwert}

\item For any two subsets $M$ and $N$ of a real vector space and for any scalar $\alpha \in \mathbb{R}$, we denote
\begin{align*}
&M + N := \{ m + n \mid m \in M, n \in N \} \,, \\
&\alpha M := \{ \alpha m \mid m \in M \} \,.
\end{align*}

\item A nonempty subset $K$ of a real vector space $V$ is a \emph{cone} in $V$ if the following conditions are satisfied:
\begin{align*}
&K + K = K \,, \\
&\alpha K \subseteq K \quad \forall \alpha \geq 0 \,, \\
&K \cap (-K) = \{ 0 \} \,.
\end{align*}
A cone $K$ in $V$ is called \emph{generating} if $K - K = V$. 

\item For a vector space $V$, let $V^*$ denote the dual space\footnote{We will only consider finite-dimensional vector spaces $V$, 
for which there is no difference between the algebraic and the topological dual space.} of $V$. A linear functional $f \in V^*$ on a 
real vector space with cone $K$ is called \emph{strictly positive} if $f(v) > 0$ for all $v \in K \setminus \{0\}$.

\xdef\letzterwert2{\the\value{enumi}}
\end{enumerate}
In the following, we recall the basic definitions in connection with abstract state spaces.
\begin{enumerate}[(Def. 1)]
\setcounter{enumi}{\letzterwert2}

\item An \emph{abstract state space} is a tuple $(A, A_+, u_A)$, where $A$ is a finite-dimensional real vector space, $A_+$ is a closed\footnote{The closedness is to be understood with respect to any norm on $A$. Since all norms on finite-dimensional vector spaces are equivalent and therefore induce the same topology, the choice of the norm is irrelevant.} and generating cone in $A$ and $u_A$ is a strictly positive linear functional, called the \emph{unit effect}. We will often denote an abstract state space by $A$ rather than $(A, A_+, u_A)$. \label{ass-def}

\item For an abstract state space $A$, the symbol $\Omega_A$ denotes the \emph{set of normalized states},
\begin{align*}
\Omega_A = \{ \omega \in A_+ \mid u_A(\omega) = 1 \} \,.
\end{align*}
The \emph{set of subnormalized states} $\Omega_A^{\leq 1}$ is defined by
\begin{align*}
\Omega_A^{\leq 1} = \{ \omega \in A_+ \mid u_A(\omega) \leq 1 \} \,.
\end{align*}
Obviously, $\Omega_A^{\leq 1} = \conv(\Omega_A \cup \{0\})$ (also see Fig. \ref{states-visualization}).
\label{omega-a-def}

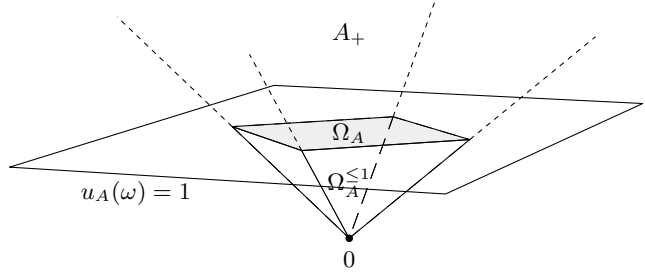
\begin{figure}[htb]
\centering

\begin{pspicture}[showgrid=false](-4.5,-0.5)(4,3.3)
\psset{viewpoint=26 10 5,Decran=40}
%\axesIIID[showOrigin=false](1,1,1)(3,2,2.5)
\psset{solidmemory}
\psSolid[object=new, linewidth=0.5\pslinewidth,
action=draw*,
name=B,
%fillcolor=red!50,
fcol=4 (.15 setfillopacity Gray),
sommets= 
%n=4
0 0 0 %0
0.420448 0.420448 0.5 %e1
-0.420448 0.420448 0.5 %e2
-0.420448 -0.420448 0.5 %e3
0.420448 -0.420448 0.5 %e4
0 0 1 %u
0 1.18921 1 %w1
-1.18921 0 1 %w2
0 -1.18921 1 %w3
1.18921 0 1, %w4
faces={
[0 7 6]
[0 8 7]
[0 9 8]
[0 6 9]
[6 7 8 9]
}]%
\psSolid[object=line, args=0 1.18921 1 0 2.37841 2, linestyle=dotted]
\psSolid[object=line, args=-1.18921 0 1 -2.37841 0 2, linestyle=dotted]
\psSolid[object=line, args=0 -1.18921 1 0 -2.37841 2, linestyle=dotted]
\psSolid[object=line, args=1.18921 0 1 2.37841 0 2, linestyle=dotted]
\psPoint(0,0,1.7){omegas}
\psdots[dotsize=0](omegas)
\uput[u](omegas){$A_+$}
\psPoint(0,0,0){zero}
\psdots[dotsize=0.1](zero)
\uput[d](zero){$0$}
\rput[b](0,0.55){$\Omega_A^{\leq 1}$}
\psPoint(0,0,0.78){omega}
\psdots[dotsize=0](omega)
\uput[u](omega){$\Omega_A$}
\psPoint(4.7,0,0.95){f}
\psdots[dotsize=0](f)
\uput[u](f){$u_A(\omega) = 1$}
\psSolid[object=plan, linewidth=0.5\pslinewidth,
definition=equation, 
args={[0 0 1 -1]}, 
base=-3 3 -2 2,action=draw,name=awesome]
\end{pspicture}

\caption{The states of any normalization are given by a cone $A_+$ in the real vector space $A$. The linear functional $u_A$ gives the normalization of a state, so the intersection of $A_+$ with the plane described by $u_A(\omega) = 1$ gives the normalized states, while the subnormalized states $\Omega_A^{\leq 1}$ are those elements of $A_+$ where $u_A$ takes values between 0 and 1.}
\label{states-visualization}
\end{figure}

\item An \emph{effect} $f$ on an abstract state space $A$ is a linear functional $f \in A^*$ such that $0 \leq f(\omega) \leq 1$ for all $\omega \in \Omega_A$. The set of effects on $A$ is denoted by $E_A$. An effect is said to be \emph{pure} if $f$ is an extreme point of $E_A$. \label{effect-def}

\item A \emph{measurement} on an abstract state space $A$ is a set $\mathcal{M} = \{ f_1, \ldots, f_n \}$ of effects which sum up to the
unit effect, $\sum_{i=1}^n f_i = u_A$. (We give this definition for the sake of completeness. We will not use this notion below, but we will 
formulate all statements in terms of effects.)

\item For a pure effect $f \in E_A$, the effect $\overline f := u_A - f$ is the \emph{complementary effect to} $f$. \label{compl-effect}

\item For a pure effect $f$, the \emph{certain face} $F_f$ of $f$ is defined by
\begin{align*}
F_f := \{ \omega \in \Omega_A \mid f(\omega) = 1 \} \,.
\end{align*}
The \emph{impossible face} $\overline F_f$ of $f$ is defined by
\begin{align*}
\overline F_f := \{ \omega \in \Omega_A \mid f(\omega) = 0 \} \,.
\end{align*}
Obviously, $\overline F_f = F_{\overline f}$. (Fact \ref{u-f-pure}) and (Fact \ref{certain-indeed-face}) below imply that $F_f$ and $\overline F_f$ are faces indeed, so the namings are justified.
\label{cert-face-def}

\item A \emph{transformation} $T: A \rightarrow A$ on an abstract state space $A$ is a map which fulfills the following conditions:
\begin{align}
&T \text{ is linear,} \nonumber \\
&T \text{ is positive, i.e. } T(A_+) \subseteq A_+, \label{trafo2} \\
&u_A(T(\omega)) \leq 1 \text{ for all } \omega \in \Omega_A. \label{trafo3}
\end{align}
Given that $T$ is linear, conditions (\ref{trafo2}) and (\ref{trafo3}) can be summarized as
\begin{align}
T(\Omega_A) \subseteq \Omega_A^{\leq 1} \,. \label{t-cond-given-lin}
\end{align}
\label{trafo-def}

\item An abstract state space $A$ is a \emph{polytopic theory} if $\Omega_A$ is a polytope. (This is what we called a discrete theory, but in this technical appendix, we use the term ``polytopic theory'' because a polytope is a well-established mathematical term in the context of convex sets.) \label{polytopic-theory}

\item An abstract state space $A$ is a \emph{classical theory} if $\Omega_A$ is a simplex. \label{classical-theory}

\xdef\letzterwert3{\the\value{enumi}}
\end{enumerate}

\subsection{Known facts}
\label{fact-section}

Recall the definition of a face and of an extreme point of a convex set, (Def. \ref{face-def}) and (Def. \ref{ext-point-def}).
\begin{enumerate}[(F{a}ct 1)]

\item For a convex set $C$, a face $S$ of a face $F$ of $C$ is itself a face of $C$. Moreover, if $S, F \subseteq C$ are faces of $C$ with $S \subseteq F$, then $S$ is a face of $F$ as well. Thus, for a face $F$ of a convex set $C$, $e \in F$ is an extreme point of $F$ (Def. \ref{ext-point-def}) if and only if it is an extreme point of $C$. \label{face-face-face}

\item For a face $F$ of a convex set $C$, if a convex combination $\sum_{i=1}^n \alpha_i v_i$ lies in $F$ for some $v_1, \ldots, v_n \in C$ and nonzero coefficients $\alpha_i$, then $v_1, \ldots, v_n \in F$. \cite[Prop. 2.7]{Pfi12} \label{nonbin-extremality}

\item If $F$ is a face of a convex set $C$, then $F = \aff(F) \cap C$. \cite[Prop. 2.10]{Pfi12} \label{face-cap-aff}

\xdef\factlastvalue{\the\value{enumi}}
\end{enumerate}
In the following, $V$ denotes a finite-dimensional real vector space and $S \subseteq V$ is any subset of $V$.
\begin{enumerate}[(F{a}ct 1)]
\setcounter{enumi}{\factlastvalue}

\item $\spa(S) = \aff(S \cup \{0\})$. \label{spa-is-aff-0}

\item If $0 \notin \aff(S)$, then $\dim(\spa(S)) = \dim S + 1$. \label{dim-span-plus-one}

\item $\aff(\conv(S)) = \aff(S)$.

\xdef\factlastthing{\the\value{enumi}}
\end{enumerate}
Let $T: A \rightarrow B$ be a linear map between vector spaces $A$ and $B$, let $S_1, S_2 \subseteq A$ be any subsets of $A$.
\begin{enumerate}[(F{a}ct 1)]
\setcounter{enumi}{\factlastthing}

\item $T(\aff(S_1 \cup S_2)) = \aff(T(S_1) \cup T(S_2))$. \label{t-aff-pres}

\item $T(\conv(S_1 \cup S_2)) = \conv(T(S_1) \cup T(S_2))$.  \label{t-conv-pres}

\item If $T$ is injective, $S_1 \subseteq S_2$ and $T(S_1) \supseteq T(S_2)$, then $S_1 = S_2$. \label{image-superset}

\xdef\lastfactvalue{\the\value{enumi}}
\end{enumerate}
We will make use of the following properties of abstract state spaces. Recall (Def. \ref{ass-def}), (Def. \ref{omega-a-def}) and (Def. \ref{effect-def}).
\begin{enumerate}[(F{a}ct 1)]
\setcounter{enumi}{\lastfactvalue}

\item Let $A$ be an abstract state space. If $f \in E_A$ is a pure effect, then the complementary effect $\overline f := u_A - f$ is a pure effect as well. \cite[Prop. 3.33]{Pfi12} \label{u-f-pure}

\item For every nonzero pure effect $f \in E_A$, the certain face $F_f$ of $f$ (as defined in (Def. \ref{cert-face-def})) is indeed a nonempty face of $\Omega_A$. \cite[Corollary 3.37]{Pfi12} Thus, by (Fact \ref{u-f-pure}), the impossible face $\overline F_f = F_{\overline f}$ is a face of $\Omega_A$ as well (which is nonempty if $f \neq u_A$). \label{certain-indeed-face}

\item For any subset $S \subseteq \Omega_A$ it holds that $0 \notin \aff(S)$. (This follows from $\aff(S) \subseteq \aff(\Omega_A) = \{ \omega \in A \mid u_A(\omega) =~1\} \not\ni 0$.) \label{0-notin-aff}

\xdef\lastfactvaluestrich{\the\value{enumi}}
\end{enumerate}
The following facts about polytopes will be useful.
\begin{enumerate}[(F{a}ct 1)]
\setcounter{enumi}{\lastfactvaluestrich}

\item Every polytope $P$ has a minus-face. (More than that, there are lower bounds on the number of minus-faces of polytopes; see \cite[Chapter 3.1]{Gru67}.)\label{polytope-has-facet}

\item A face of a polytope is a polytope. (This is easily proved using (Fact \ref{nonbin-extremality}).) \label{polytope-face-lemma}

\xdef\lastfactvaluestrichlein{\the\value{enumi}}
\end{enumerate}

\subsection{Technical lemmas}
\label{tech-lemma-section}

In this section, we prove four technical lemmas. Before we prove them, we give a rough overview over their role in the main proof.
The main result, Theorem \ref{main-result-thm}, is proved in two steps:
\begin{enumerate}[(i)]
\item First, we show that every polytopic theory $A$ which satisfies our postulate has a set of normalized states $\Omega_A$ which is uniformly pyramidal (we will see in Section \ref{main-result-section} what this means).
\item In the second step, we show that every uniformly pyramidal polytope is a simplex, so $A$ is a classical theory.
\end{enumerate}
The main technical lemma that allows us to prove these two steps is Lemma \ref{main-technical-lemma}. It establishes geometric criteria that a set of states has to satisfy to be compatible with our postulate by specifying conditions on the certain face (Def. \ref{cert-face-def}) of pure effects. This lemma has two parts (a) and (b), both of which we will use in step (i) of the proof of Theorem~\ref{main-result-thm}. In order to prove Lemma \ref{main-technical-lemma}, we need Lemma \ref{face-pres-lemma} which we prove first (this lemma will also be useful in the proof of step (ii)). To be applied properly, Lemma \ref{main-technical-lemma} needs a helper which comes in the form of Lemma \ref{unique-suitable-f}. It shows that in the case where $\Omega_A$ is a polytope, every minus-face of $\Omega_A$ is the certain face of a pure effect. Lemma \ref{u_s-face-lemma} in turn helps us to prove Lemma \ref{unique-suitable-f}. Figure \ref{lemma-structure} gives an overview over the organisation of the proofs.

\begin{figure}[htb]
\centering

\begin{pspicture}[showgrid=false](-1,1)(8.5,5)
\psset{linewidth=0.7\pslinewidth, nodesep=3pt}
\psnode(4,3.7){l1}{\psframebox{Lemma \ref{face-pres-lemma}}}
\psnode(7,4.2){l2}{\psframebox{Lemma \ref{main-technical-lemma}}}
\psnode(6.4,3.7){l2a}{(a)}
\psnode(7.5,3.7){l2b}{(b)}
\psnode(6.4,2){thmi}{(i)}
\psnode(7.5,2){thmii}{(ii)}
\psnode(4,2){l4}{\psframebox{Lemma \ref{unique-suitable-f}}}
\psnode(0.5,2){l3}{\psframebox{Lemma \ref{u_s-face-lemma}}}
\psnode(7,1.5){l2}{\psframebox{Theorem \ref{main-result-thm}}}
\ncline{->}{l1}{l2a}
\ncangle[angleA=-90,angleB=90,armB=0.6cm,linearc=.2]{->}{l1}{thmii}
\ncline{->}{l2a}{thmi}
\ncline{->}{l4}{thmi}
\ncline{->}{l2b}{thmi}
\ncline{->}{l3}{l4}
\ncline{->}{thmi}{thmii}
\end{pspicture}

\caption{Organization of the proof of the main result. This diagram shows how the proof of Theorem \ref{main-result-thm} is subdivided into several Lemmas.}
\label{lemma-structure}
\end{figure}
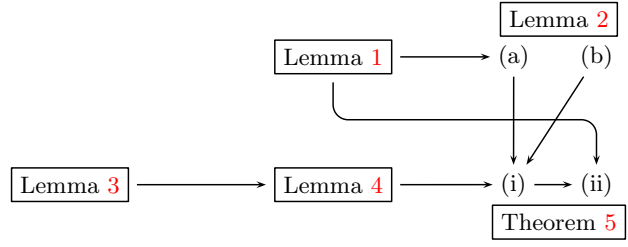

Now we prove the four lemmas.

\begin{lemma}
\label{face-pres-lemma}
Let $C$ be a convex subset of a vector space $V$, let $a \in V$ be affinely independent of $C$ (i.e. $a \notin \aff(C)$). If $F$ is a face of $C$, then $D = \conv(F \cup \{a\})$ is a face of $\conv(C \cup \{a\})$.
\end{lemma}

\begin{figure}[htb]
\centering

\begin{pspicture}[showgrid=false](-1.2,0.2)(1.3,3.4) 
\psset{viewpoint=10 10 25 rtp2xyz,Decran=11} 
\psset{solidmemory}
\psSolid[
object=new, linewidth=0.7\pslinewidth,
linecolor=gray,
%fcol=0 (0.5 setfillopacity Gray),
name=B,
sommets=
0 0 3 %0 0
-0.707107 1.22474 1 %w1 1
-0.707107 -1.22474 1 %w2 2
1.41421 0 1, %w3 3
faces={
[1 2 3]
[1 2 0]
[2 3 0]
[3 1 0]
},
action=draw*]%
\psSolid[object=plan,definition=solidface,action=none,args=B 2,name=R0]
\psset{fontsize=20}
\psset{phi=90}
\psProjection[object=texte,text={D},plan=R0]%
\psSolid[object=plan,definition=solidface,action=none,args=B 3,name=R1]
\psset{fontsize=20}
\psset{phi=0}
\psProjection[object=texte,text={C},plan=R1]%
\pstThreeDPut(1.41421, 1.6, 4.2){$F$}
\psSolid[object=line, linewidth=4\pslinewidth,args=1.41421 0 1 0 0 2.98]
\psSolid[object=line, linewidth=4\pslinewidth,args=1.41421 0.015 1 0 0.015 2.965]
\psPoint(-0.707107, -1.22474, 1){p}
\uput[l](p){$a$}
\psdots[dotsize=0.1](p)
\end{pspicture}

\caption{Visualization of Lemma 1: For a convex set $C$, a face $F \subseteq C$ and a point $a \notin \aff(C)$, 
									the set $D = \conv(F \cup \{a\})$ is a face of $\conv(C \cup \{a\})$.}
\label{face-pres-figure}
\end{figure}
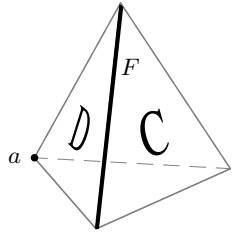

\begin{proof}
Let $x, y \in \conv(C \cup \{a\})$. Then,
\begin{align*}
x = \sum\limits_{i=1}^{m-1} \alpha_i p_i + \alpha_m a &\qquad \text{for some } \alpha_i \geq 0 \text{ with } \sum\limits_{i=1}^m \alpha_i = 1 \\
&\qquad \text{and } p_i \in C \ \forall i \in \{1, \ldots, m-1\} \,.
\end{align*}
We can simplify the expression for $x$ by defining
\begin{align*}
&p := \frac{\sum_{i=1}^{m-1} \alpha_i p_i}{\sum_{i=1}^{m-1} \alpha_i} \in C, \quad \alpha := \sum\limits_{i=1}^{m-1} \alpha_i, \quad \overline \alpha := \alpha_m
\end{align*}
to get
\begin{align}
&x = \alpha p + \overline \alpha a \quad \text{with} \quad \alpha, \overline \alpha \geq 0 \,, \quad \alpha + \overline \alpha = 1\,, \quad p \in C \,. \label{x-eq}
\end{align}
In a similar way, we get
\begin{align}
&y = \beta q + \overline \beta a \quad \text{for some} \quad \beta, \overline \beta \geq 0 \,, \quad \beta + \overline \beta = 1 \,, \quad q \in C \,. \label{y-eq}
\end{align}
Let $0 < \lambda < 1$, $\overline \lambda = 1-\lambda$. Then,
\begin{align}
\label{xy-eq}
\lambda x + \overline \lambda y = \lambda \alpha p + \overline \lambda \beta q + (\lambda \overline \alpha + \overline \lambda \, \overline \beta) a \,.
\end{align}
Suppose that
\begin{align}
\label{xyz-eq}
\lambda x + \overline \lambda y = z \quad \text{for some} \quad z \in \conv(F \cup \{a\}) \,.
\end{align}
According to (Def. \ref{face-def}), the claim is proved if we manage to show that (\ref{xyz-eq}) implies that $x, y \in \conv(F \cup \{a\})$. The point $z$ being an element of $\conv(F \cup \{a\})$ means that
\begin{align}
&z = \mu r + \overline \mu a \quad \text{for some} \quad \mu, \overline \mu \geq 0 \,, \quad \mu + \overline \mu = 1 \,, \quad r \in F \,. \label{z-eq}
\end{align}
The case $\mu = 0$ is easy: In this case, $\lambda x + \overline \lambda y = a$, but $a$ is an extreme point of $\conv(C \cup \{a\})$ and thus $x = y = a \in \conv(F \cup \{a\})$. Therefore, we consider the more difficult case and assume that
\begin{align}
\mu \neq 0 \,. \label{mu-not-0}
\end{align}
Equations (\ref{xy-eq}), (\ref{xyz-eq}) and (\ref{z-eq}) give
\begin{align}
\label{xyz-impl}
\mu r + \overline \mu a = \lambda \alpha p + \overline \lambda \beta q + (\lambda \overline \alpha + \overline \lambda \, \overline \beta) a \,.
\end{align}
Now \emph{assume} that
\begin{align}
\label{coeff-ass}
\overline \mu - \lambda \overline \alpha - \overline \lambda \, \overline \beta \neq 0 \,.
\end{align}
This assumption allows us to rewrite (\ref{xyz-impl}) in the form
\begin{align}
\label{aff-comb-for-a}
a = \frac{1}{\overline \mu - \lambda \overline \alpha - \overline \lambda \, \overline \beta} (\lambda \alpha p + \overline \lambda \beta q - \mu r) \,.
\end{align}
It is easily checked that the right hand side of (\ref{aff-comb-for-a}) is an affine combination of $p, q$ and $r$:
\begin{align*}
\frac{\lambda \alpha + \overline \lambda \beta - \mu}{\overline \mu - \lambda \overline \alpha - \overline \lambda \, \overline \beta}
= \frac{\lambda \alpha + (1-\lambda)\beta - \mu}{(1-\mu) - \lambda (1-\alpha) - (1-\lambda)(1-\beta)} = 1 \,.
\end{align*}
Thus, assumption (\ref{coeff-ass}) implies that $a \in \aff( \{p, q, r\} ) \subseteq \aff(C)$. This contradicts the premise that $a \notin \aff(C)$, so assumption (\ref{coeff-ass}) must be wrong and therefore
\begin{align}
\overline \mu = \lambda \overline \alpha + \overline \lambda \, \overline \beta \label{omu-eq}
\end{align}
Equation (\ref{omu-eq}) simplifies Equation (\ref{xyz-impl}) to
\begin{align}
\mu r = \lambda \alpha p + \overline \lambda \beta q \label{xyz-simpl}
\end{align}
Writing out $\overline \mu = 1-\mu$, $\overline \alpha = 1-\alpha$ and $\overline \beta = 1 - \beta$, it is easily checked that Equation (\ref{omu-eq}) implies
\begin{align}
\frac{\lambda \alpha}{\mu} + \frac{\overline \lambda \beta}{\mu} = 1 \,, \quad \text{where } \mu \neq 0 \text{ by (\ref{mu-not-0})} \label{frac-eq} \,.
\end{align}
We can rewrite (\ref{xyz-simpl}) as
\begin{align}
r = \frac{\lambda \alpha}{\mu} p + \frac{\overline \lambda \beta}{\mu} q \quad \text{with} \quad p, q \in C, \quad r \in F \,. \label{r-eq}
\end{align}
Equations (\ref{frac-eq}) and (\ref{r-eq}), together with the fact that $F$ is a face of $C$, implies that $p, q \in F$ (c.f. (Def. \ref{face-def})). Thus, by Equations (\ref{x-eq}) and (\ref{y-eq}), we have that $x, y \in \conv(F \cup \{a\})$, which completes the proof.
\end{proof}

Before we state and prove the next lemma, we want to point out a difference from previous versions of this article. Previously, Lemma \ref{main-technical-lemma} (b) did not state that it is sufficient for us to analyze the case where $\overline F_f$ consists of a single point. As this is the only case we need later, we had implicitly assumed this here, and without this restriction, Lemma~\ref{main-technical-lemma} (b) would not hold. We now corrected this mistake in our statement.

\begin{lemma}
\label{main-technical-lemma}
Let $(A, A_+, u_A)$ be an abstract state space, let $f \in E_A$ be a pure effect. If there exists a transformation $T: A \rightarrow A$ such that $u_A \circ T = f$ and $T(\omega) = \omega$ for every $\omega \in F_f$ (Def. \ref{cert-face-def}), then
\begin{enumerate}[(a)]
\item $\dim F_f + \dim \overline F_f \leq \dim \Omega_A - 1$ and
\item if $\overline F_f$ consists of not more than one point, then \\ $\aff(F_f \cup \overline F_f) \cap \Omega_A = \conv(F_f \cup \overline F_f)$.
\end{enumerate}
\end{lemma}

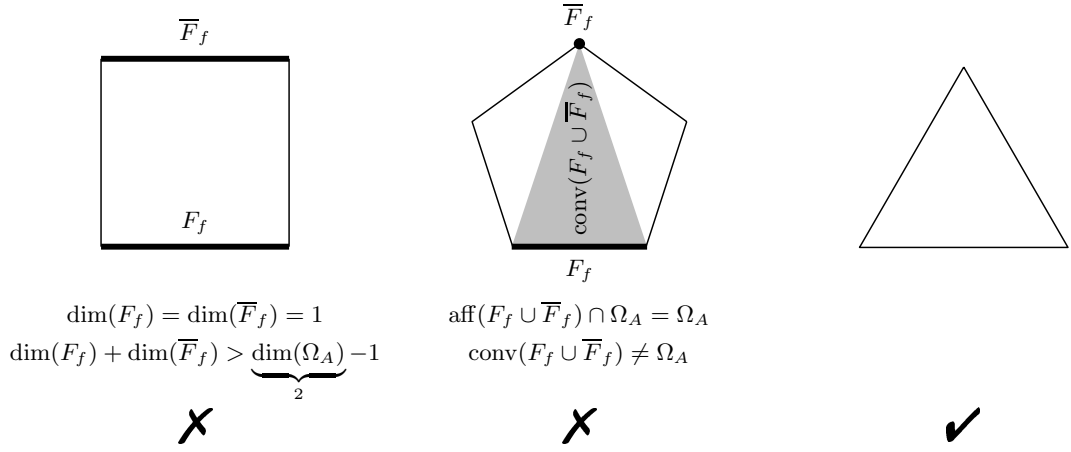
\begin{figure}[htb]
\centering	

\psset{linewidth=0.7\pslinewidth}
\begin{pspicture}[showgrid=false](-3,-4)(3.5,1.5)
\psline[linewidth=4\pslinewidth](-1.25,-1.25)(1.25,-1.25)
\psline[linewidth=4\pslinewidth](-1.25,1.25)(1.25,1.25)
\psline(-1.25,-1.25)(-1.25,1.25)
\psline(1.25,-1.25)(1.25,1.25)
\uput[90](0,-1.25){$F_f$}
\uput[90](0,1.25){$\overline F_f$}
\uput[270](0,-1.8){$\dim(F_f)=\dim(\overline F_f)=1$}
\uput[270](0,-2.3){$\dim(F_f) + \dim(\overline F_f) > \underbrace{\dim(\Omega_A)}_{2} - 1$}
\uput[270](0,-3.2){\Large \XSolidBrush}
\end{pspicture}
\begin{pspicture}[showgrid=false](-1.5,-4)(3.5,1.5)
\pspolygon[linestyle=none, fillstyle=solid, fillcolor=lightgray](-0.9,-1.25)(0,1.45)(0.9,-1.25)
\rput(0,-0.05){\PstPentagon[PstPicture=false, unit=1.5]}
\psline[linewidth=4\pslinewidth](-0.9,-1.25)(0.9,-1.25)
%\psline[linestyle=dashed](-0.9,-1.25)(0,1.45)
%\psline[linestyle=dashed](0.9,-1.25)(0,1.45)
\psdot[dotsize=0.15](0,1.45)
\uput[270](0,-1.25){$F_f$}
\uput[90](0,1.45){$\overline F_f$}
\rput[l]{90}(-0.01,-1.1){$\conv(F_f \cup \overline F_f)$}
\uput[270](0,-1.8){$\aff(F_f \cup \overline F_f) \cap \Omega_A = \Omega_A$}
\uput[270](0,-2.3){$\conv(F_f \cup \overline F_f) \neq \Omega_A$}
\uput[270](0,-3.2){\Large \XSolidBrush}
\end{pspicture}
\begin{pspicture}[showgrid=false](-1.5,-4)(3,1.5)
\rput(0,-0.46){\PstTriangle[PstPicture=false, unit=1.6]}
\uput[270](0,-3.2){\Large \CheckmarkBold}
\end{pspicture}

\caption{A few examples illustrating Lemma \ref{main-technical-lemma}: The square violates condition (a). As we have demonstrated in the main text, this leads to a dimension mismatch for the transformation. The pentagon satisfies condition (a) but violates condition (b). In the main text, we have seen that this results in a non-positive transformation: The white parts of the pentagon (i.e. the parts outside of the gray area $\conv(F_f \cup \overline F_f)$) are mapped outside of $\Omega_A^{\leq 1}$. The triangle satisfies both (a) and (b).}
\label{main-technical-visualization}
\end{figure}

\begin{proof}
Let $(A, A_+, u_A)$ be an abstract state space, let $f \in E_A$ be a pure effect and $T: A \rightarrow A$ be a transformation with $u_A \circ T = f$ and $T(\omega) = \omega$ for every $\omega \in F_f$. For the rest of the proof, it is useful to write out (Def. \ref{trafo-def}) of a transformation and to list all the properties of $T$:
\begin{align}
&T: A \rightarrow A \text{ map such that} \nonumber \\
&\qquad T \text{ is linear,} \label{t-lin} \\
&\qquad T \text{ is positive, i.e. } T(A_+) \subseteq A_+ \,, \label{t-pos} \\
&\qquad u_A \circ T = f \label{t-ind-f} \\
&\qquad \qquad \Rightarrow u_A(T(\omega)) \leq u_A(\omega) \ \forall \omega \in \Omega_A \,, \label{t-norm-nonincr} \\
&\qquad T(\omega) = \omega \ \forall \omega \in F_f \,. \label{t-post-4}
\end{align}
Our goal is to show that properties (\ref{t-lin} - \ref{t-post-4}) imply (a) and (b) as stated above.

First note that (\ref{t-post-4}) and the linearity of $T$ (\ref{t-lin}) imply
\begin{align}
T|_{\spa(F_f)} = I|_{\spa(F_f)} \label{I-F-f} \,,
\end{align}
where $I|_{\spa(F_f)}$ is the restriction of the identity operator to $\spa(F_f)$. On the other hand, by the definition of $\overline F_f$, it holds that $f(\omega) = 0$ for all $\omega \in \overline F_f$, so by (\ref{t-ind-f}), we have that $u_A(T(\omega)) = 0$ for all $\omega \in \overline F_f$. This implies $T(\omega)$ is the zero-vector for all $\omega \in \overline F_f$, so by the linearity of $T$ (\ref{t-lin}), this means that
\begin{align}
T|_{\spa(\overline F_f)} = 0|_{\spa(\overline F_f)} \label{0-overline-F-f} \,,
\end{align}
where $0|_{\spa(\overline F_f)}$ denotes the restriction of the zero-operator to $\spa(\overline F_f)$.

\begin{enumerate}[(a)]

\item For Equations (\ref{I-F-f}) and (\ref{0-overline-F-f}) to be satisfied simultaneously, we must have that
\begin{align*}
\spa(F_f) \cap \spa(\overline F_f) = \{0\} \,,
\end{align*}
which is only possible if
\begin{align}
\dim(\spa(F_f)) + \dim(\spa(\overline F_f)) \leq \dim A \label{dim-ineq-prev}
\end{align}
Since $0 \notin \aff(F_f)$ and $0 \notin \aff(\overline F_f)$ (Fact \ref{0-notin-aff}), (Fact \ref{dim-span-plus-one}) implies that
\begin{align*}
&\dim(\spa(F_f)) = \dim(F_f) + 1 \,, \\
&\dim(\spa(\overline F_f)) = \dim(\overline F_f) + 1 \,.
\end{align*}
Noting that $\dim A = \dim \Omega_A + 1$, this allows us to rewrite Inequality (\ref{dim-ineq-prev}):
\begin{align*}
\dim F_f + \dim \overline F_f + 2 \leq \dim \Omega_A + 1
\end{align*}
and therefore
\begin{align*}
\dim F_f + \dim \overline F_f \leq \dim \Omega_A - 1 \,.
\end{align*}

\item Equations (\ref{I-F-f}) and (\ref{0-overline-F-f}) imply that
\begin{align*}
&T(F_f) = F_f \,, \\
&T(\overline F_f) = \{0\} \,.
\end{align*}
Thus, by (Fact \ref{t-conv-pres}), it holds that
\begin{align*}
T(\conv(F_f \cup \overline F_f)) &= \conv(T(F_f) \cup T(\overline F_f)) \\
&= \conv(F_f \cup \{0\}) \,.
\end{align*}
By Lemma \ref{face-pres-lemma}, the set $\conv(F_f \cup \{0\})$ is a face of $\conv(\Omega_A \cup \{0\}) = \Omega_A^{\leq 1}$. This allows us to apply (Fact \ref{face-cap-aff}) to see that
\begin{align*}
\conv(F_f \cup \{0\}) &= \underbrace{\aff(\conv(F_f \cup \{0\}))}_{\aff(F_f \cup \{0\}) \text{ by (Fact \ref{t-aff-pres})}} \cap \ \Omega_A^{\leq 1} \\
&= \aff(F_f \cup \{0\}) \cap \Omega_A^{\leq 1}
\end{align*}
and thus
\begin{align}
T(\conv(F_f \cup \overline F_f)) = \aff(F_f \cup \{0\}) \cap \Omega_A^{\leq 1} \label{t-conv} \,.
\end{align}
In the following, we show that this contains $T(\aff(F_f \cup \overline F_f) \cap \Omega_A)$. First note that
\begin{align*}
&T(\aff(F_f \cup \overline F_f) \cap \Omega_A) \subseteq T(\aff(F_f \cup \overline F_f)) \cap T(\Omega_A) \,.
\end{align*}
We can rewrite this term by means of (Fact \ref{t-aff-pres}),
\begin{align}
&T(\aff(F_f \cup \overline F_f) = \aff(T(F_f) \cup T(\overline F_f)) \,, \nonumber
\end{align}
and by means of (\ref{t-pos}) and (\ref{t-norm-nonincr}) (c.f. (\ref{t-cond-given-lin})),
\begin{align}
&T(\Omega_A) \subseteq \Omega_A^{\leq 1} \,,
\end{align}
to get
\begin{align}
&T(\aff(F_f \cup \overline F_f) \cap \Omega_A) \subseteq \aff(\underbrace{T(F_f)}_{F_f} \cup \underbrace{T(\overline F_f)}_{\{0\}}) \cap \underbrace{T(\Omega_A)}_{\subseteq \Omega_A^{\leq 1}} \,.\nonumber
\end{align}
Thus,
\begin{align}
&T(\aff(F_f \cup \overline F_f) \cap \Omega_A) \subseteq \aff(F_f \cup \{0\}) \cap \Omega_A^{\leq 1} \label{t-aff-subset} \,.
\end{align}
Note that in the assumed case where $\overline F_f$ consists of only one point, $T|_{\aff(F_f \cup \overline F_f)}$ is injective. (This is easily verified from Equations (\ref{I-F-f}) and (\ref{0-overline-F-f}): $T|_{\aff(F_f)}$ is injective, $\overline F_f$ is affinely independent of $F_f$ and $T(F_f) = F_f$ is affinely independent of $T(\overline F_f) = \{ 0 \}$, so the affine map $T|_{\aff(F_f \cup \overline F_f)}$ is injective.) Moreover,
\begin{align}
\conv(F_f \cup \overline F_f) \subseteq \aff(F_f \cup \overline F_f) \cap \Omega_A \,. \label{nice-subset}
\end{align}
By virtue of (Fact \ref{image-superset}), Equations (\ref{t-conv}), (\ref{t-aff-subset}) and (\ref{nice-subset}) imply
\begin{align*}
\aff(F_f \cup \overline F_f) \cap \Omega_A = \conv(F_f \cup \overline F_f) \,,
\end{align*}
which is what we wanted to show. \qedhere
\end{enumerate}

\end{proof}

\begin{lemma}
\label{u_s-face-lemma}
Let $A$ be an abstract state space, let $S \subseteq \Omega_A$ be any subset of the normalized states. Then
\begin{align*}
U_S := \{ f \in E_A \mid f(\omega) = 1 \ \forall \omega \in S \}
\end{align*}
is a face of $E_A$.
\end{lemma}

\begin{figure}[htb]
\centering

\begin{pspicture}[showgrid=false](-2,-0.5)(3,3)
\psset{viewpoint=20 50 10,Decran=120}
%\axesIIID[showOrigin=false](1,1,1)(3,2,2.5)
\psset{solidmemory}
\psSolid[object=new, linewidth=0.5\pslinewidth, fcol=true,
action=draw*,
name=C,
%fillcolor=red!50,
fcol=3 (0.75 setgray),
sommets=
%n=3
0 0 0 %0 0
-0.707107 1.22474 1 %w1 1
-0.707107 -1.22474 1 %w2 2
1.41421 0 1 %w3 3
-0.235702 0.408248 0.333333 %e1 4
-0.235702 -0.408248 0.333333 %e2 5
0.471405 0 0.333333 %e3 6
0.235702 -0.408248 0.666667 %u-e1 7
0.235702 0.408248 0.666667 %u-e2 8
-0.471405 0 0.666667 %u-e3 9
0 0 1, %u 10
faces={
[0 4 8 6]
[0 6 7 5]
[0 5 9 4]
[7 6 8 10]
[8 4 9 10]
[9 5 7 10]
%[1 2 3]
}]%
\psSolid[object=new, linewidth=0.5\pslinewidth,
action=draw,
%name=C,
%fillcolor=red!50,
%fcol=3 (Gray),
sommets=
%n=3
0 0 0 %0 0
-0.707107 1.22474 1 %w1 1
-0.707107 -1.22474 1 %w2 2
1.41421 0 1 %w3 3
-0.235702 0.408248 0.333333 %e1 4
-0.235702 -0.408248 0.333333 %e2 5
0.471405 0 0.333333 %e3 6
0.235702 -0.408248 0.666667 %u-e1 7
0.235702 0.408248 0.666667 %u-e2 8
-0.471405 0 0.666667 %u-e3 9
0 0 1, %u 10
faces={
[1 2 3]
}]%
\psPoint(1.41421, 0, 1){w3}
\uput[l](w3){$S$}
\psdots[dotsize=0.12](w3)
\psPoint(0, 0, 0.85){us}
\uput[dl](us){$U_S$}
\end{pspicture}
\begin{pspicture}[showgrid=false](-2,-0.5)(3,3)
\psset{viewpoint=26 10 5,Decran=60}
%\axesIIID[showOrigin=false](1,1,1)(3,2,2.5)
\psset{solidmemory}
\psSolid[object=new, linewidth=0.5\pslinewidth, linecolor=gray,
action=draw*,
name=B,
%fillcolor=red!50,
%fcol=8 (.5 setfillopacity Blue),
sommets= 
%n=4
0 0 0 %0
0.420448 0.420448 0.5 %e1
-0.420448 0.420448 0.5 %e2
-0.420448 -0.420448 0.5 %e3
0.420448 -0.420448 0.5 %e4
0 0 1 %u
0 1.18921 1 %w1
-1.18921 0 1 %w2
0 -1.18921 1 %w3
1.18921 0 1, %w4
faces={
[1 4 0]
[0 2 1]
[0 3 2]
[0 4 3]
[1 2 5]
[2 3 5]
[3 4 5]
[4 1 5]
[6 7 8 9]}]%
\psSolid[object=line, linewidth=2\pslinewidth, args=1.18921 0 1 0 -1.18921 1]
\psSolid[object=line, linewidth=2\pslinewidth, args=1.18921 0 1.006 0 -1.18921 1.006]
\psSolid[object=line, linewidth=2\pslinewidth, args=1.18921 0 0.994 0 -1.18921 0.994]
\psSolid[object=line, linewidth=2\pslinewidth, args=0.420448 -0.420448 0.5 0 0 1]
\psSolid[object=line, linewidth=2\pslinewidth, args=0.420448 -0.420448 0.51 0 0 1.01]
\psSolid[object=line, linewidth=2\pslinewidth, args=0.420448 -0.420448 0.49 0 0 0.99]
\psPoint(-0.2, -1.18921, 0.9){w3p}
\uput[dr](w3p){$S$}
\psPoint(0.40448, -0.35448, 0.62){e4}
\uput[u](e4){$U_S$}
\end{pspicture}
\begin{pspicture}[showgrid=false](-2,-0.5)(1.5,3)
\psset{viewpoint=26 10 5,Decran=60}
%\axesIIID[showOrigin=false](1,1,1)(3,2,2.5)
\psset{solidmemory}
\psSolid[object=new, linewidth=0.5\pslinewidth,
action=draw*,
name=A,
%fillcolor=red!50,
fcol=10 (0.75 setgray),
sommets= 
%n=5
0 0 0 %0
0.152217 0.468477 0.414214 %e1
-0.402248 0.29225 0.447214 %e2
-0.402248 -0.29225 0.447214 %e3
0.153645 -0.472871 0.447214 %e4
0.497206 0 0.447214 %e5
-0.153645 -0.472871 0.552786 %u-e1
0.402248 -0.29225 0.552786 %u-e2
0.402248 0.29225 0.552786 %u-e3
-0.153645 0.472871 0.552786 %u-e4
-0.497206 0 0.552786 %u-e5
0 0 1 %u
0.343561 1.05737 1 %w1
-0.899454 0.653491 1 %w2
-0.899454 -0.653491 1 %w3
0.343561 -1.05737 1 %w4
1.11179 0 1, %w5
faces={
[0 1 8 5]
[0 5 7 4]
[0 4 6 3]
[0 3 10 2]
[0 2 9 1]
[6 4 7 11]
[7 5 8 11]
[8 1 9 11]
[9 2 10 11]
[10 3 6 11]
[12 13 14 15 16]}]%
\psPoint(0, 0, 1){u}
\psdots[dotsize=0.12](u)
\uput[r](u){$U_S = \{u_A\}$}
\uput[dr](-1,3){$S = \Omega_A$}
\psline{->}(-0.8,2.55)(-0.7,2.1)
\end{pspicture}

\caption{Illustration of Lemma \ref{u_s-face-lemma}: This figure shows three examples where the set $U_S$ is visualized. As one can see, in all three cases, the set $U_S$ is a face of $E_A$. For simplicity of the picture, we have chosen the subset $S \subseteq \Omega_A$ to be a face of $\Omega_A$. However, the fact that $U_S$ is a face of $E_A$ also holds when $S$ is not a face of $\Omega_A$. In the square example, for instance, if we would take any subset of the edge $S$ with at least two elements, then the set $U_S$ would still be the same.}
\label{u_s-examples}
\end{figure}
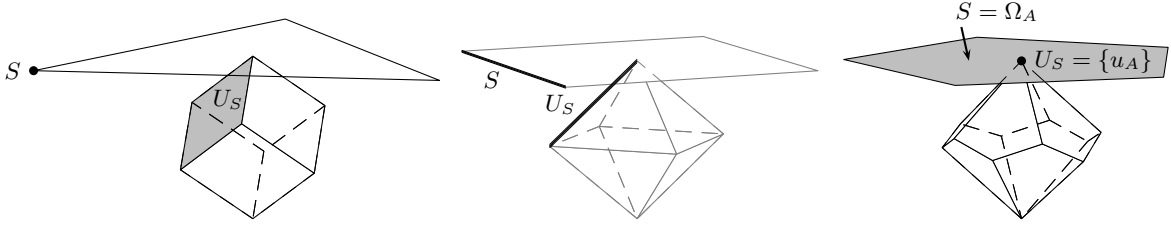

\begin{proof}
We have to check the properties listed in (Def. \ref{face-def}). Obviously, $u_A \in U_S$, so $U_S$ is nonempty. If $f_1, f_2 \in E_A$ with $f_1(\omega) = f_2(\omega) = 1$ for all $\omega \in S$, then $\lambda f_1(\omega) + (1-\lambda) f_2(\omega) = 1$ for all $\lambda \in [0,1]$ and all $\omega \in S$, so $U_S$ is convex. Let $f \in U_S$, let $g, h \in E_A$ and $0 < \lambda < 1$ such that $\lambda g + (1-\lambda) h = f$. For any $\omega \in S$, we have that
\begin{align}
\underbrace{\lambda}_{< 1} \underbrace{g(\omega)}_{\leq 1} + \underbrace{(1-\lambda)}_{<1} \underbrace{h(\omega)}_{\leq 1} = 1 \,. \label{fgh-eq}
\end{align}
Equation (\ref{fgh-eq}) can only be satisfied if $g(\omega) = h(\omega) = 1$, so $g, h \in U_S$.
\end{proof}

\begin{lemma}
\label{unique-suitable-f}
Let $A$ be an abstract state space, let $F$ be a minus-face of $\Omega_A$ (Def. \ref{minus-face}). Then there is a unique pure effect $f \in E_A$ such that $F$ is the certain face of $f$ (Def. \ref{cert-face-def}), i.e. $F_f = F$.
\end{lemma}

\begin{proof}
First, note that any effect $g$ with $F_g = F$ must be an element of $U_F = \{ g \in E_A \mid g(\omega) = 1 \ \forall \omega \in F \}$ since $g \in U_F$ is equivalent to $F \subseteq F_g$. Let $g \in U_F$. By the linearity of $g$, the condition
\begin{align}
\label{g-condition}
g(\omega) = 1 \ \forall \omega \in F
\end{align}
determines $g$ on $\spa(F)$. We know from (Fact \ref{0-notin-aff}) and (Fact \ref{dim-span-plus-one}) that $\dim(\spa(F)) = \dim F + 1$. Moreover, the premise that $F$ is a minus-face of $\Omega_A$ gives $\dim F + 1 = \dim \Omega_A - 1 + 1 = \dim \Omega_A = \dim A - 1$ and thus $\dim(\spa(F)) = \dim A - 1$. Thus, any functional $g \in A^*$ satisfying condition (\ref{g-condition}) is fully determined by specifying its value at some point $p \notin \spa(F)$. Let $\alpha, \beta \in \mathbb{R}$, $\alpha \neq \beta$, and let $g_\alpha, g_\beta \in A^*$ be the unique functional satisfying (\ref{g-condition}) and $g_\alpha(p) = \alpha, g_\beta(p) = \beta$, respectively. Any $g \in A^*$ satisfying (\ref{g-condition}) lies in the affine hull of $g_\alpha$ and $g_\beta$ since for $g \in A^*$ with $g(p) = \gamma$, it holds that
\begin{align*}
g(p) = \gamma &= \frac{\gamma - \beta}{\alpha - \beta} \alpha + \left(1 - \frac{\gamma - \beta}{\alpha - \beta}\right) \beta = \frac{\gamma - \beta}{\alpha - \beta} g_\alpha(p) + \left(1 - \frac{\gamma - \beta}{\alpha - \beta}\right) g_\beta(p) \,.
\end{align*}
Thus,
\begin{align*}
U_F &= \{ g \in E_A \mid g(\omega) = 1 \ \forall \omega \in F \} \subseteq \aff(\{g_\alpha\} \cup \{g_\beta\}) \\
&\Rightarrow \dim U_F = 1 \,,
\end{align*}
where we have used the fact that for a minus-face $F$, the set $U_F$ contains more elements than just $u_A$.\footnote{This can be seen geometrically: Regarding $\aff(\Omega_A)$ as an affine space, $\Omega_A$ fits between two parallel affine hyperplanes (in $\aff(\Omega_A)$) such that one of them touches $\Omega_A$ at $F$. Define an affine functional with value 1, 0 on the hyperplane, respectively. The linear extension of this functional to $A$ is an effect which is different from $u_A$. This proof sketch can be turned into a rigorous and elementary but lenghty proof.} According to Lemma \ref{u_s-face-lemma}, $U_F$ is a face of $E_A$. Therefore, $U_F$ is a convex 1-dimensional set, which is nothing but a line-segment. A line segment has exactly two extreme points, namely its endpoints. By (Fact \ref{face-face-face}), these two extreme points are precisely the extreme points of $E_A$ compatible with (\ref{g-condition}) and therefore pure effects. The pure effect $u_A \in E_A$ is obviously one of these two pure effects (for which $F_{u_A} = \Omega_A$, see Fig. \ref{three-ex-fig}). Let the other one be denoted by $f$. Obviously, $F \subseteq F_f$ since $f \in U_F$. By (Fact \ref{certain-indeed-face}), $F_f$ is a face of $\Omega_A$. However, by the premise that $F$ is a minus-face of $\Omega_A$ and by (Fact \ref{face-cap-aff}), the only faces of $\Omega_A$ containing $F$ are $F$ and $\Omega_A$. The latter can be excluded since $f \neq u_A$. Thus, $f$ is the unique pure effect such that $F_f = F$.
\end{proof}

\begin{figure}[htb]
\centering

\begin{pspicture}[showgrid=false](-2,-1.5)(3,3)
\psset{viewpoint=26 10 5,Decran=60}
%\axesIIID[showOrigin=false](1,1,1)(3,2,2.5)
\psset{solidmemory}
\psSolid[object=new, linewidth=0.5\pslinewidth, linecolor=gray,
action=draw*,
name=B,
%fillcolor=red!50,
%fcol=8 (.5 setfillopacity Blue),
sommets= 
%n=4
0 0 0 %0
0.420448 0.420448 0.5 %e1
-0.420448 0.420448 0.5 %e2
-0.420448 -0.420448 0.5 %e3
0.420448 -0.420448 0.5 %e4
0 0 1 %u
0 1.18921 1 %w1
-1.18921 0 1 %w2
0 -1.18921 1 %w3
1.18921 0 1, %w4
faces={
[1 4 0]
[0 2 1]
[0 3 2]
[0 4 3]
[1 2 5]
[2 3 5]
[3 4 5]
[4 1 5]
[6 7 8 9]}]%
\psSolid[object=line, linewidth=2\pslinewidth, args=1.18921 0 1 0 -1.18921 1]
\psSolid[object=line, linewidth=2\pslinewidth, args=1.18921 0 1.006 0 -1.18921 1.006]
\psSolid[object=line, linewidth=2\pslinewidth, args=1.18921 0 0.994 0 -1.18921 0.994]
\psSolid[object=line, linewidth=2\pslinewidth, args=0.420448 -0.420448 0.5 0 0 1]
\psSolid[object=line, linewidth=2\pslinewidth, args=0.420448 -0.420448 0.51 0 0 1.01]
\psSolid[object=line, linewidth=2\pslinewidth, args=0.420448 -0.420448 0.49 0 0 0.99]
\psPoint(-0.2, -1.18921, 0.9){w3p}
\uput[dr](w3p){$F$}
\psPoint(0.40448, -0.35448, 0.62){e4}
\uput[u](e4){$U_F$}
\psPoint(0, 0, 1){u}
\psdots[dotsize=0.13](u)
\uput[r](u){$u_A$}
\psPoint(0.420448, -0.420448, 0.5){e4}
\psdots[dotsize=0.13](e4)
\uput[l](e4){$f$}
\uput[270](0,-0.4){$\left.\begin{array}{l}\text{The only elements of $U_F$ that} \\ \text{are pure effects are $u_A$ and $f$.}\end{array}\right.$}
\end{pspicture}
\begin{pspicture}[showgrid=false](-2,-1.5)(4,3)
\psset{viewpoint=26 10 5,Decran=60}
%\axesIIID[showOrigin=false](1,1,1)(3,2,2.5)
\psset{solidmemory}
\psSolid[object=new, linewidth=0.5\pslinewidth,
action=draw*,
name=B,
%fillcolor=red!50,
fcol=8 (0.75 setgray),
sommets= 
%n=4
0 0 0 %0
0.420448 0.420448 0.5 %e1
-0.420448 0.420448 0.5 %e2
-0.420448 -0.420448 0.5 %e3
0.420448 -0.420448 0.5 %e4
0 0 1 %u
0 1.18921 1 %w1
-1.18921 0 1 %w2
0 -1.18921 1 %w3
1.18921 0 1, %w4
faces={
[1 4 0]
[0 2 1]
[0 3 2]
[0 4 3]
[1 2 5]
[2 3 5]
[3 4 5]
[4 1 5]
[6 7 8 9]}]%
\psPoint(0, 0, 1){u}
\psdots[dotsize=0.13](u)
\uput[r](u){$u_A$}
\uput[dr](-1,3){$F_{u_A} = \Omega_A \neq F$}
\psline{->}(-0.8,2.55)(-0.7,2.1)
\uput[270](0,-0.4){$\left.\begin{array}{l}\text{The certain face of } u_A \\ \text{does \emph{not} coincide with } F.\end{array}\right.$}
\end{pspicture}
\begin{pspicture}[showgrid=false](-2,-1.5)(3,3)
\psset{viewpoint=26 10 5,Decran=60}
%\axesIIID[showOrigin=false](1,1,1)(3,2,2.5)
\psset{solidmemory}
\psSolid[object=new, linewidth=0.5\pslinewidth, linecolor=gray,
action=draw*,
name=B,
%fillcolor=red!50,
%fcol=8 (.5 setfillopacity Blue),
sommets= 
%n=4
0 0 0 %0
0.420448 0.420448 0.5 %e1
-0.420448 0.420448 0.5 %e2
-0.420448 -0.420448 0.5 %e3
0.420448 -0.420448 0.5 %e4
0 0 1 %u
0 1.18921 1 %w1
-1.18921 0 1 %w2
0 -1.18921 1 %w3
1.18921 0 1, %w4
faces={
[1 4 0]
[0 2 1]
[0 3 2]
[0 4 3]
[1 2 5]
[2 3 5]
[3 4 5]
[4 1 5]
[6 7 8 9]}]%
\psSolid[object=line, linewidth=2\pslinewidth, args=1.18921 0 1 0 -1.18921 1]
\psSolid[object=line, linewidth=2\pslinewidth, args=1.18921 0 1.006 0 -1.18921 1.006]
\psSolid[object=line, linewidth=2\pslinewidth, args=1.18921 0 0.994 0 -1.18921 0.994]
\psPoint(0.420448, -0.420448, 0.5){e4}
\psdots[dotsize=0.13](e4)
\uput[l](e4){$f$}
\psPoint(-0.2, -1.18921, 0.9){w3p}
\uput[d](w3p){$F_f = F$}
\uput[270](0,-0.4){$\left.\begin{array}{l}\text{The certain face of } f \\ \text{\emph{does} coincide with } F.\end{array}\right.$}
\end{pspicture}

\caption{\label{three-ex-fig}The basic idea behind the proof of Lemma \ref{unique-suitable-f}: Here we illustrate the case where $F$ is a minus-face of a square-shaped set of states.}
\end{figure}
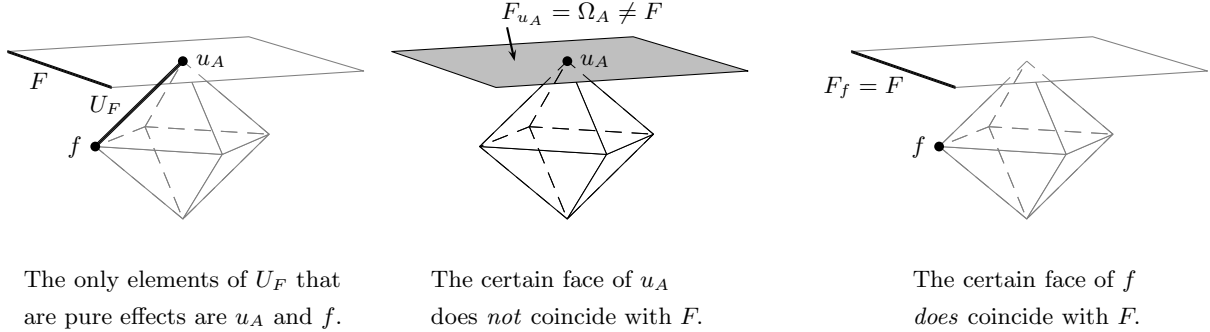

\subsection{The main result}
\label{main-result-section}

The following definition will be useful in the proof of Theorem \ref{main-result-thm}.

\begin{defi}
We call a polytope $P$ \emph{uniformly pyramidal} (For a motivation of this naming, see Example 2.40 and Definition 5.3 in \cite{Pfi12}.) if for every minus-face $F$ of $P$, it holds that $P = \conv(F \cup \{a_F\})$ for some $a_F \in P$. Note that in this case, it obviously holds that $a_F \notin \aff(F)$.
\end{defi}

\begin{figure}[htb]
\centering

\begin{pspicture}[showgrid=false](-2,-1.5)(6,4) 
\psset{viewpoint=10 10 30 rtp2xyz,Decran=11} 
\psset{solidmemory}
\psSolid[
object=new, linewidth=0.5\pslinewidth,
%fcol=0 (0.5 setfillopacity Gray),
name=B,
sommets=
0 0 3 %0 0
-0.707107 1.22474 1 %w1 1
-0.707107 -1.22474 1 %w2 2
1.41421 0 1, %w3 3
faces={
[1 2 3]
[1 2 0]
[2 3 0]
[3 1 0]
},
action=draw*]%
\psSolid[
object=new, linewidth=0.5\pslinewidth,
%fcol=0 (0.5 setfillopacity Gray),
name=C,
sommets=
0 0 3 %0 0
-0.707107 1.22474 1 %w1 1
-0.707107 -1.22474 1 %w2 2
1.41421 0 1, %w3 3
faces={
[2 1 0]
},
action=none]%
\psSolid[object=plan,definition=solidface,action=none,args=B 2,name=R0]
\psset{fontsize=15}
\psset{phi=90}
\psProjection[object=texte,text={B},plan=R0]%
\psSolid[object=plan,definition=solidface,action=none,args=B 3,name=R1]
\psset{fontsize=15}
\psset{phi=0}
\psProjection[object=texte,text={E},plan=R1]%
\psSolid[object=plan,definition=solidface,action=none,args=B 0,name=R2]
\psset{fontsize=15}
\psset{phi=270}
\psProjection[object=texte,text={C},plan=R2]%
\psSolid[object=plan,definition=solidface,action=none,args=C 0,name=R3]
\psset{fontsize=15}
\psset{phi=180}
\psProjection[object=texte,text={D},plan=R3]%
\psPoint(-0.707107, -1.22474, 1){p}
\uput[l](p){$a_E$}
\psPoint(-0.707107, 1.22474, 1){r}
\uput[r](r){$a_B$}
\psPoint(1.41421, 0, 1){s}
\uput[d](s){$a_D$}
\psPoint(0, 0, 3){t}
\uput[u](t){$a_C$}
\rput[tl](2.5, 2.5){\let\centering\relax\parbox{3cm}{$T = \conv(F \cup \{a_F\})$ for every minus-face $F$, i.e. for $F = B, C, D, E$}}
\uput[d](0,-0.2){\parbox{4cm}{\centering tetrahedron $T$ \\ uniformly pyramidal}}
\psdots[dotsize=0.1](p)(r)(s)(t)
\end{pspicture}
\begin{pspicture}[showgrid=false](-2.5,-1)(6,4)
\psset{viewpoint=40 16 12,Decran=130}
%\axesIIID[showOrigin=false](1,1,1)(3,2,2.5)
\psset{solidmemory}
\psSolid[object=new,linewidth=0.5\pslinewidth,
action=draw*,
name=A,
sommets= 
%n=4
0 0 0 %0
0.420448 0.420448 0.5 %e1
-0.420448 0.420448 0.5 %e2
-0.420448 -0.420448 0.5 %e3
0.420448 -0.420448 0.5 %e4
0 0 1.2 %u
0 1.18921 1 %w1
-1.18921 0 1 %w2
0 -1.18921 1 %w3
1.18921 0 1, %w4
faces={
[1 2 3 4]
[1 2 5]
[2 3 5]
[3 4 5]
[4 1 5]}]%
\psSolid[object=plan,definition=solidface,action=none,args=A 0,name=S0]
\psset{fontsize=8}
\psset{phi=0}
\psProjection[object=texte,text={G},plan=S0]%
\psSolid[object=plan,definition=solidface,action=none,args=A 1,name=S1]
\psset{fontsize=8}
\psset{phi=90}
\psProjection[object=texte,text={H},plan=S1]%
\psPoint(0, 0, 1.2){a}
\uput[u](a){$a_G$}
\psdots[dotsize=0.1](a)
\rput[tl](2, 3){\let\centering\relax\parbox{3cm}{$P = \conv(G \cup \{a_G\})$, but there is no \\ $a_H \in P$ such that $P = \conv(H \cup \{a_H\})$}}
\uput[d](0,0.3){\parbox{4cm}{\centering pyramid $P$ \\ \emph{not} uniformly pyramidal}}
\end{pspicture}

\caption{The uniformal pyramidal property: The tetrahedron $T$ on the left is an example of a uniformly pyramidal polytope. The shape $P$ on the right (formed like an Egyptian pyramid) is not uniformly pyramidal: it is only pyramidal with respect to its ground face $G$. For every other face $F$, there are two extreme points of the pyramid that are not contained in the face $F$, so the pyramid is not of the form $\conv(F \cup \{a_F\})$.}
\label{unif-pyr-figure}
\end{figure}
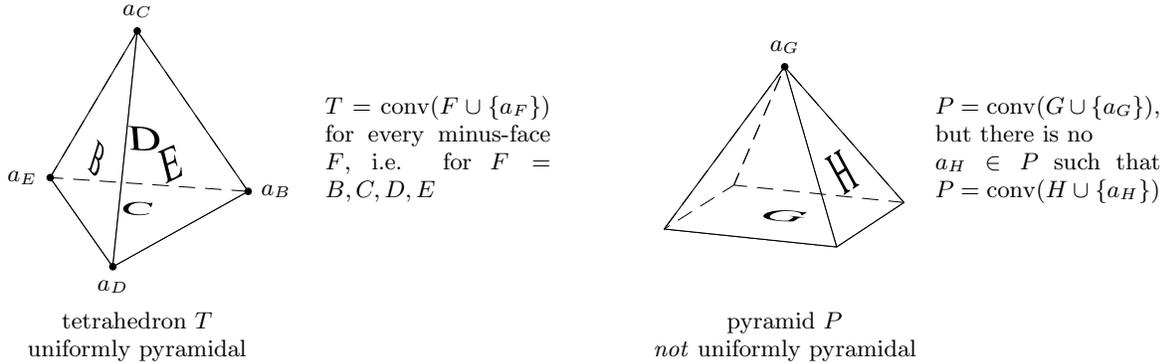

From a physical point of view, it would be sufficient to derive the result from the assumption that our postulate holds for \emph{every} pure effect. However, in Theorem \ref{main-result-thm}, we assume even less: We only assume the postulate for pure effects for which the certain face $F_f$ is a minus-face of $\Omega_A$. This is a weaker assumption and thus, we prove a stronger statement. This will be useful in Appendix \ref{approx-case-section}. To prove Theorem \ref{approx-case-thm}, we will make use of Lemma \ref{non-pres-f}, which is the contraposition of Theorem \ref{main-result-thm}. This contraposition has the right form if we only assume the postulate for pure effects for which the certain face $F_f$ is a minus-face of $\Omega_A$.

\begin{thm}
\label{main-result-thm}
Let $A$ be a polytopic theory (Def. \ref{polytopic-theory}) satisfying the following weak form of our postulate: For every pure effect $f \in E_A$ for which the certain face $F_f$ is a minus-face of $\Omega_A$, there is a transformation $T: A \rightarrow A$ such that $f = u_A \circ T$ and $T(\omega) = \omega$ for every $\omega \in F_f$ (Def. \ref{cert-face-def}). Then $A$ is a classical theory (Def. \ref{classical-theory}).
\end{thm}

\begin{proof}
We prove this theorem in two steps.
\begin{enumerate}[(i)]
\item In the first step, we show that the assumptions imply that the polytope $\Omega_A$ is uniformly pyramidal (see the definition above).
\item Then we show that a uniformly pyramidal polytope $\Omega_A$ must be a simplex, so $A$ is a classical theory.
\end{enumerate}
Now we prove each of the two steps.
\begin{enumerate}[(i)]

\item Let $F$ be a minus-face of $\Omega_A$ (which exists by (Fact \ref{polytope-has-facet})). By (Def. \ref{minus-face}), this means that $\dim F = \dim \Omega_A - 1$. We have proved in Lemma \ref{unique-suitable-f} that there is a unique pure effect $f \in E_A$ such that $F$ is the certain face of $f$, i.e. $F_f = F$ (Def. \ref{cert-face-def}). Let $T: A \rightarrow A$ be a transformation such that $f = u_A \circ T$ and $T(\omega) = \omega$ for every $\omega \in F_f$. By Lemma \ref{main-technical-lemma} (a), we have that
\begin{align*}
\dim F_f + \dim \overline F_f &= \dim \Omega_A - 1 + \dim \overline F_f \\
&\leq \dim \Omega_A - 1 \\
&\Rightarrow \dim \overline F_f \leq 0 \,.
\end{align*}

\begin{figure}
\centering

\begin{pspicture}[showgrid=false](-6,-2)(6,2)
\psset{linewidth=0.7\pslinewidth}
\psline[linewidth=3\pslinewidth](-1.5,-1)(1.5,-1)
\psline(-1.5,-1)(0,1)
\psline(0,1)(1.5,-1)
\psline[linestyle=dashed](-1.5,-1)(-4,-1)
\psline[linestyle=dashed](1.5,-1)(4,-1)
\uput[l](-4,-1){$\aff(F_f)$}
\uput[d](0,-1){$F = F_f$}
\psdots[dotsize=0.1](0,1)
\uput[u](0,1){$\overline F_f = \{a_F\}$}
\uput[d](0,0){$\Omega_A$}

\end{pspicture}

\caption{The geometrical entities discussed in step (i) of the proof: Here we see the case where $\Omega_A$ is two-dimensional. The set $F$ is a minus-face of the states $\Omega_A$. It is the certain face of a pure effect $f \in E_A$, i.e. $F_f = f$. The impossible face $\overline F_f$ of $f$ consists of a single point $a_F$. It holds that $\aff(F_f) \cap \overline F_f = \emptyset$, from which it follows that $\aff(F_f \cup \overline F_f) \cap \Omega_A = \Omega_A$ since $F_f$ is a minus-face of $\Omega_A$.}
\label{step-i-illustration}
\end{figure}
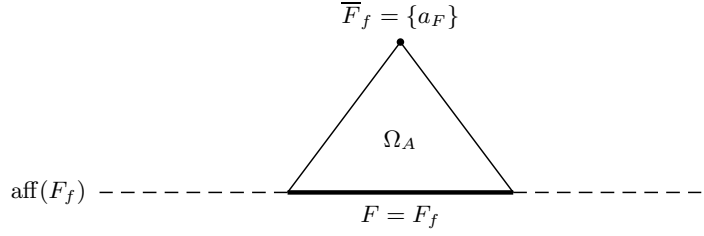

Thus, $\overline F_f$ must be a singleton or the empty set (Def. \ref{set-dim-def}). The latter is excluded since the pure effect $f$ is different from $u_A$ (because $F_{u_A} = \Omega_A \neq F$). Therefore, $\overline f$ (Def. \ref{compl-effect}) is nonzero and pure (Fact \ref{u-f-pure}) and thus $\overline F_f = F_{\overline f}$ is nonempty (Fact \ref{certain-indeed-face}). Thus, $\overline F_f = \{a_F\}$ for some $a_F \in \Omega_A$. Note that $\aff(F_f) \cap \overline F_f = \emptyset$ since $\aff(F_f) \cap \Omega_A = F_f$ and $F_f \cap \overline F_f = \emptyset$. Thus, $\aff(F_f \cup \overline F_f) \cap \Omega_A = \Omega_A$ since $\dim F_f = \dim \Omega_A - 1$ and $\overline F_f \subseteq \Omega_A \setminus \aff(F_f)$. Thus, by Lemma \ref{main-technical-lemma} (b), we have that
\begin{align*}
\Omega_A &=\aff(F_f \cup \overline F_f) \cap \Omega_A \\
&= \conv(F \cup \{a_F\}) \,.
\end{align*}
The point $a_F$ is affinely independent of $F$ since $\aff(F_f) \cap \overline F_f = \emptyset$, as we have already noticed.

\item Now we show that every uniformly pyramidal polytope $\Omega_A$ is a simplex. We prove this by induction over $\dim \Omega_A$. The base case $\dim \Omega_A = 0$ is trivial: A singleton is a simplex. The case $\dim \Omega_A = 1$ is equally easy: Every one-dimensional polytope is a line-segment, and a line-segment is a simplex.

Let $\Omega_A$ be a uniformly pyramidal polytope with $\dim \Omega_A = d \geq 2$. Assume that every $(d-1)$-dimensional uniformly pyramidal polytope is a simplex (induction hypothesis). Let $F \subseteq \Omega_A$ be a minus-face of $\Omega_A$. Since $\Omega_A$ is uniformly pyramidal, we have that $\Omega_A = \conv(F \cup \{a_F\})$ for some $a_F \in \Omega_A$ with $a_F \notin \aff(F)$. We want to show that $\Omega_A$ is a simplex. To this end, it is sufficient to show that $F$ is a simplex since $F$ being a simplex and $\Omega_A = \conv(F \cup \{a_F\})$ with $a_F \notin \aff(F)$ implies that $\Omega_A$ is a simplex (c.f. (Def. \ref{simplex-def})). We show that $F$ is uniformly pyramidal (by the induction hypothesis, this implies that $F$ is a simplex).

\begin{figure}[htb]
\centering

\begin{pspicture}[showgrid=false](-1.2,0.2)(1.3,3.4) 
\psset{viewpoint=10 10 25 rtp2xyz,Decran=11} 
\psset{solidmemory}
\psSolid[
object=new, linewidth=0.7\pslinewidth,
linecolor=gray,
fcol=0 (Gray),
name=B,
sommets=
0 0 3 %0 0
-0.707107 1.22474 1 %w1 1
-0.707107 -1.22474 1 %w2 2
1.41421 0 1, %w3 3
faces={
[1 2 3]
[1 2 0]
[2 3 0]
[3 1 0]
},
action=draw*]%
\psSolid[object=plan,definition=solidface,action=none,args=B 2,name=R0]
\psset{fontsize=20}
\psset{phi=90}
\psProjection[object=texte,text={H},plan=R0]%
\psSolid[object=plan,definition=solidface,action=none,args=B 3,name=R1]
\psset{fontsize=20}
\psset{phi=0}
\psProjection[object=texte,text={F},plan=R1]%
\psPoint(-0.707107, -1.22474, 1){p}
\uput[l](p){$a_F$}
\pstThreeDPut(1.41421, 1.6, 4.2){$G$}
\pstThreeDPut(0, 1.6, 1.4){$\Omega_A$}
\psSolid[object=line, linewidth=4\pslinewidth,args=1.41421 0 1 0 0 2.98]
\psSolid[object=line, linewidth=4\pslinewidth,args=1.41421 0.015 1 0 0.015 2.965]
\psPoint(-0.707107, 1.22474, 1){r}
\uput[r](r){$a_H$}
\psdots[dotsize=0.1](p)(r)
\end{pspicture}

\caption{Step (ii) of the proof: This figure visualizes the definitions in the proof that every uniformly pyramidal polytope $\Omega_A$ is a simplex.}
\label{simplex-proof-figure}
\end{figure}
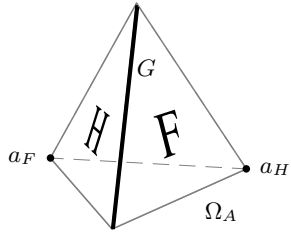

The set $F$ is a face of a polytope and therefore a polytope itself (Fact \ref{polytope-face-lemma}). Recalling (Def. \ref{number-ext-points}), we see that the equation $\Omega_A = \conv(F \cup \{a_F\})$ implies that
\begin{align}
n_e(\Omega_A) = n_e(F) + 1 \label{count-vertices-1} \,.
\end{align}
Let $G$ be a minus-face of $F$ (see Fig. \ref{simplex-proof-figure}). By Lemma \ref{face-pres-lemma}, $H := \conv(G \cup \{a_F\})$ is a face of $\conv(F \cup \{a_F\}) = \Omega_A$. The dimension of $H$ is given by $\dim H = \dim(\conv(G \cup \{a_F\})) = \dim G + 1 = \dim F = \dim \Omega_A - 1$, so $H$ is a minus-face of $\Omega_A$. Thus, since $\Omega_A$ is uniformly pyramidal, $\Omega_A = \conv(H \cup \{a_H\})$ for some $a_H \in \Omega_A$. This allows us to see that
\begin{align*}
n_e(G) &= n_e(\conv(G) \cup \{a_F\}) -1 \\
&= n_e(H) - 1 \\
&= n_e(\conv(H \cup \{a_H\})) - 2 \\
&= n_e(\Omega_A) - 2 \,,
\end{align*}
so by (\ref{count-vertices-1}),
\begin{align}
\label{count-vertices-2}
n_e(G) = n_e(F) - 1 \,.
\end{align}
Equation (\ref{count-vertices-2}) and the fact that $G$ is a minus-face of $F$ imply that
\begin{align}
F = \conv(G \cup \{b_G\}) \quad \text{for some } b_G \in F \,. \label{f-unif-pyr}
\end{align}
Since $G$ is an arbitrary minus-face of $F$, (\ref{f-unif-pyr}) implies that $F$ is uniformly pyramidal. By the induction hypothesis, it follows that $F$ is a simplex, so $\Omega_A = \conv(F \cup \{a_F\})$ is a simplex since $a_F$ is affinely independent of $F$. \qedhere

\end{enumerate}
\end{proof}

\clearpage
\section{Formal proof of the approximate case}
\label{approx-case-section}

In the following, we will give a formal prove of the approximate version of the result. ``Approximate'' means that this version shows that discrete (polytopic) non-classical theories are ruled out even if the postulate is weakened to an approximate version. Therefore, the result presented here is \emph{stronger} than the original version of the result. The statement of the approximate version is more difficult to read than the original version. It reads:
\begin{quote}
Let $A$ be a polytopic non-classical theory (Def. \ref{polytopic-theory}) and (Def. \ref{classical-theory}) and let $\Vert \cdot \Vert_A$ be any norm on $A$. Then there is a pure effect $f \in E_A$ and a positive number $\epsilon$ with the following property: For every transformation $T: A \rightarrow A$ with $f = u_A \circ T$, there is a state $\rho \in \Omega_A$ with $f(\rho) = 1$ and $\Vert T(\rho) - \rho \Vert_A \geq \epsilon$.
\end{quote}
This statement has the form of a contrapositive of the original version. To see why this is a stronger version of the main result, the reader is invited to convince himself that if $\epsilon$ is chosen to be just zero (instead of being positive) and the distance $\Vert T(\rho) - \rho \Vert_A \geq \epsilon$ is replaced by $\Vert T(\rho) - \rho \Vert_A > 0$, then this statement boils down to the original version of the result. The statement above (with positive $\epsilon$) means that a polytopic non-classical theory cannot even satisfy our postulate ``up to $\epsilon$''.

As before, we organise the proof in a concise way by first listing all the definitions that we need in Section \ref{not-and-def-b} and by listing all the facts that we will use in the proof in Section \ref{known-facts-b}. This time, the lists are shorter since we do not repeat definitions and facts of Appendix \ref{main-result-appendix}. In Section \ref{technical-lemmas-b}, we give an overview over the idea behind the proof and prove all the lemmas that we will need. Finally, we give the proof of the approximate version in Section \ref{approx-proof-section}.

\subsection{Notation and Definitions}
\label{not-and-def-b}

\begin{enumerate}[(Def. 1)]
\setcounter{enumi}{\letzterwert3}

\item For a normed space $A$ and a closed subset $C \subseteq A$, we define
\begin{align*}
\left.\begin{array}{cccl}d( \, \cdot \, , C): & A & \rightarrow & \mathbb{R} \\ & x & \mapsto & d(x, C) := \inf\limits_{y \in C} \Vert y - x \Vert\end{array}\right.
\end{align*}
This map has the property that for all $x \in A$, it holds that $d(x, C) \geq 0$ with equality if and only if $x \in C$.
\label{set-distance-def}

\item For a vector space $A$, let $\End(A)$ denote the space of endomorphisms on $A$, i.e. the vector space of all linear maps from $A$ to itself.

\item For an abstract state space $(A, A_+, u_A)$ and an effect $f \in E_A$, we define the set of all transformations that induce the effect $f$ as
\begin{align*}
\mathcal{T}_f := \{ T \in \End(A) \mid T \text{ is positive, } u_A \circ T = f \} \,.
\end{align*}
Note that since $T \in \End(A)$ implies linearity of $T$ and $u_A \circ T = f$ implies $u_A(T(\omega)) \leq 1$ for all $\omega \in \Omega_A$, the elements of $\mathcal{T}_f$ are precisely the transformations (Def. \ref{trafo-def}) $T: A \rightarrow A$ with $u_A \circ T = f$. \label{t-f-set-def}

If $\Vert \cdot \Vert_A$ is a norm on $A$, then the operator norm $\Vert \cdot \Vert_{\End(A)}$ induces a metric $d_{\mathcal{T}_f}(S, T) := \Vert S - T \Vert_{\End(A)}$ on $\mathcal{T}_f$, which turns $\mathcal{T}_f$ into a metric space.
\label{t-f-ind}

\item For an abstract state space $A$, an effect $f \in E_A$ and a norm $\Vert \cdot \Vert_A$ on $A$, we define the \emph{disturbance function}
\begin{align*}
\left.\begin{array}{cccl}D_f: & \mathcal{T}_f & \rightarrow & \mathbb{R} \\ & T & \mapsto & \max\limits_{\omega \in F_f} \Vert T(\omega) - \omega \Vert_A\end{array}\right.
\end{align*}
For every transformation $T: A \rightarrow A$ which induces the effect $f$ (i.e. $u_A \circ T = f$), the disturbance function evaluates the maximal disturbance on the certain face $F_f$ of $f$ (Def. \ref{cert-face-def}) caused by the transformation. \label{disturb-f}

\end{enumerate}

\subsection{Known facts}
\label{known-facts-b}

\begin{enumerate}[(F{a}ct 1)]
\setcounter{enumi}{\lastfactvaluestrichlein}

\item For any norm-induced topology on a finite-dimensional vector space $A$, a polytope $P \subseteq A$ is compact. \label{polytope-compact}

\item On a finite-dimensional vector space $A$, any two norms $\Vert \cdot \Vert_A$ and $\Vert \cdot \Vert'_A$ are equivalent, i.e. there are positive constants $c_1, c_2$ such that $c_1 \Vert v \Vert_A \leq \Vert v \Vert'_A \leq c_2 \Vert v \Vert_A$ for all $v \in A$.
\label{norm-equivalence}

\item (Heine-Borel Theorem) In a finite-dimensional normed space $A$, a subset $S \subseteq A$ is compact if and only if $S$ is closed and bounded. \label{heine-borel}

\item In a normed vector space $A$, the closure $\overline S$ of a subset $S \subseteq A$ coincides with the set of all limits of sequences in $S$ that converge in $A$. Thus, a subset $S \subseteq A$ is closed if (and only if) every sequence in $S$ that converges in $A$ has its limit in $S$. \label{closure-limits}

\item If $f: X \rightarrow Y$ is a linear map between between finite-dimensional normed spaces, then for every convergent sequence $x_n \rightarrow x \in X$, it holds that $f(x_n) \rightarrow f(x) \in Y$. \label{sequent-cont}

\end{enumerate}

\subsection{Technical lemmas}
\label{technical-lemmas-b}

In this section, we will prove four technical lemmas that will allow us to prove Theorem \ref{approx-case-thm}. To see where things are going, we first present a sketch of the organisation of the proof. As a corollary of Theorem \ref{main-result-thm} that we have proved in Appendix \ref{main-result-appendix}, we will first show in Lemma \ref{non-pres-f} that for every non-classical polytopic theory, there is a pure effect $f$ such that its certain face $F_f$ (Def. \ref{cert-face-def}) is a minus-face of the states and such that there is no transformation that induces $f$ which satisfies our postulate precisely. We start the proof of Theorem \ref{approx-case-thm} by considering a non-classical polytopic theory. We apply Lemma \ref{non-pres-f} which allows us to consider a pure effect $f$ with the mentioned properties.

For the actual proof of the theorem, we then make a distinction of cases (see Fig. \ref{lemma-structure-b}). We first prove case (i) where we assume that the dimension of the impossible face $\overline F_f$ of $f$ is zero-dimensional (in other words, consists of a single point). It is practical to consider this case separately since in this case, we can make a proof that considers a linear map $L$ (with certain properties) which does not exist if $\overline F_f$ is higher-dimensional.

The other case (ii) is the case where $\dim \overline F_f \geq 1$. Lemma \ref{t-f-dim-lemma} will help us to show that in this case, any transformation $T$ that induces $f$ (i.e. $T \in \mathcal{T}_f$, (Def. \ref{t-f-set-def})) must map $F_f$ to a set $T(F_f)$ of lower dimension than $F_f$. This implies that for every $T \in \mathcal{T}_f$, there is a state $\omega$ such that $T(\omega) \neq \omega$ and therefore $\Vert T(\omega) - \omega \Vert_A > 0$. This will show that the disturbance function $D_f$ is a positive function on $\mathcal{T}_f$. At this point, we will be left to show that $D_f$ is lower-bounded by a positive number $\epsilon$. We will show this by showing that $\mathcal{T}_f$ is compact (Lemma \ref{compactness-lemma}) and that $D_f$ is a continuous map on $\mathcal{T}_f$ (Lemma \ref{continuity-lemma}).

\begin{figure}[htb]
\centering

\begin{pspicture}[showgrid=false](-1,1)(9,5)
\psset{linewidth=0.7\pslinewidth, nodesep=3pt}
\psnode(0.5,4.5){oldthm}{\psframebox[linecolor=gray]{\textcolor{gray}{Theorem \ref{main-result-thm}}}}
\psnode(4,4.5){bl1}{\psframebox{Lemma \ref{non-pres-f}}}
\pnode(4,3.2){split}{}
\psnode(4,1.6){newthm}{\psframebox{\qquad \quad Theorem \ref{approx-case-thm} \qquad \quad}}
\psnode(3.5,2.5){newthmi}{(i)}
\psnode(4.5,2.5){newthmii}{(ii)}
\psnode(8,4){bl2}{\psframebox{Lemma \ref{t-f-dim-lemma}}}
\psnode(8,3){bl3}{\psframebox{Lemma \ref{compactness-lemma}}}
\psnode(8,2){bl4}{\psframebox{Lemma \ref{continuity-lemma}}}
\pnode(6,2.5){merge}{}
\ncline[linecolor=gray]{->}{oldthm}{bl1}
\ncline[nodesepB=0]{-}{bl1}{split}
\nlput[offset=-7pt](bl1)(split){1cm}{$f$}
\ncline[nodesepA=0]{->}{split}{newthmi}
\ncline[nodesepA=0]{->}{split}{newthmii}
\ncline{->}{bl2}{newthmii}
\ncline[nodesepB=0]{-}{bl3}{merge}
\ncline[nodesepB=0]{-}{bl4}{merge}
\ncline[nodesepA=0]{->}{merge}{newthmii}
\rput[tr](3.6,2.3){$\dim \overline F_f = 0$}
\rput[tl](4.3,2.3){$\dim \overline F_f \geq 1$}
\end{pspicture}

\caption{Organization of the proof of the approximate case: This diagram shows how the proof of Theorem \ref{approx-case-thm} is subdivided into several Lemmas.}
\label{lemma-structure-b}
\end{figure}
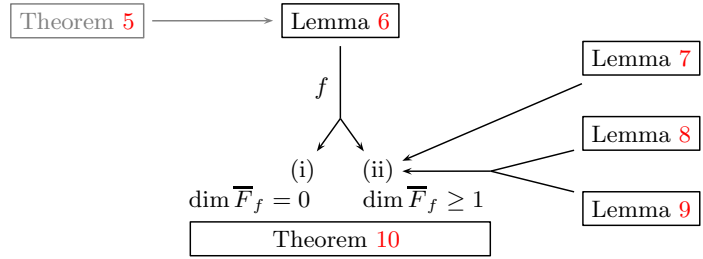

\begin{lemma}
\label{non-pres-f}
Let $A$ be a polytopic non-classical theory. Then there is a pure effect $f \in E_A$ such that the certain face $F_f$ is a minus-face of $\Omega_A$ and such that there is no transformation $L: A \rightarrow A$ with $f = u_A \circ L$ and $L(\omega) = \omega$ for all $\omega \in F_f$.
\end{lemma}

\begin{proof}
Let $\{ F_i \}_{i \in I}$ be the set of minus-faces of $\Omega_A$. By Lemma \ref{unique-suitable-f}, for every $i \in I$, there is a unique pure effect $f_i \in E_A$ such that $F_{f_i} = F_i$. Since $A$ is polytopic but non-classical, we can apply the contraposition of Theorem \ref{main-result-thm} to see that there must be a $k \in I$ such that there is no transformation $L: A \rightarrow A$ with $f_k = u_A \circ L$ and $L(\omega) = \omega$ for all $\omega \in F_{f_k}$. Thus, $f := f_k$ is the effect we were looking for.
\end{proof}

\begin{lemma}
\label{t-f-dim-lemma}
Let $A$ be a non-trivial abstract state space (i.e. $\dim A > 1$). Let $f \in E_A$ be a pure effect. Let $T: A \rightarrow A$ be a transformation such that $f = u_A \circ T$. Then, for the certain face $F_f$ of $f$ (Def. \ref{cert-face-def}), we have
\begin{align*}
\dim(T(F_f)) \leq \dim A - \dim \overline F_f - 2 \,,
\end{align*}
where $\overline F_f$ is the impossible face of $f$ (Def. \ref{cert-face-def}) (Recall the convention $\dim (\emptyset) = -1$ (Def. \ref{set-dim-def})).
\end{lemma}

\begin{proof}
The impossible face $\overline F_f$ is the subset of $\Omega_A$ where $f$ vanishes, $\overline F_f = \{ \omega \in \Omega_A \mid f(\omega) = 0 \}$. Thus, the assumption that $f = u_A \circ T$ means that $f(\omega) = 0$ for all $\omega \in \overline F_f$ implies $T(\overline F_f) = \{0\}$ (since the zero-vector is the only element of $\Omega_A^{\leq 1}$ with normalization equal to zero). By linearity of $T$, this implies $T(\spa(\overline F_f)) = \{0\}$, so $\ker(T) \supseteq \spa(\overline F_f)$ and therefore $\dim(\ker(T)) \geq \dim(\spa(\overline F_f))$. (Fact \ref{0-notin-aff}) and (Fact \ref{dim-span-plus-one}) imply that $\dim(\spa(\overline F_f)) = \dim(\overline F_f) + 1$ and thus $\dim(\ker(T)) \geq \dim(\overline F_f) + 1$. Therefore,
\begin{align}
\label{dim-im-t-a-1}
\dim(\im(T)) \leq \dim A - \dim(\overline F_f) - 1 \,.
\end{align}
On the other hand, the condition that $f(\omega) = u_A(T(\omega)) = 1$ for all $\omega \in F_f$ implies that $T(F_f) \subseteq \Omega_A$ (rather than just $T(F_f) \subseteq \Omega_A^{\leq 1}$). Therefore,
\begin{align}
\label{t-f-subset-cap}
T(F_f) \subseteq \im(T) \cap \Omega_A \subseteq \im(T) \cap \aff(\Omega_A) \,.
\end{align}
In an abstract state space, we have that
\begin{align}
\label{dim-om-a-1}
\dim(\aff(\Omega_A)) = \dim A - 1 \,.
\end{align}
Moreover,
\begin{align}
\label{0-in-im-t}
0 \in \im(T) \,, \quad \text{but} \quad 0 \notin \aff(\Omega_A) \quad \text{by (Fact \ref{0-notin-aff})} \,.
\end{align}
We can combine (\ref{dim-im-t-a-1}), (\ref{dim-om-a-1}) and (\ref{0-in-im-t}) to see that
\begin{align*}
\dim(\im(T) \cap \aff(\Omega_A)) \leq \dim A - \dim(\overline F_f) - 2 \,.
\end{align*}
Thus, by (\ref{t-f-subset-cap}),
\begin{align*}
\dim(T(F_f)) \leq \dim A - \dim(\overline F_f) - 2
\end{align*}
as claimed.
\end{proof}

\begin{lemma}
\label{compactness-lemma}
Let $A$ be an abstract state space, let $f \in E_A$ be an effect. Then, for any norm $\Vert \cdot \Vert_A$ on $A$, the space $(\mathcal{T}_f, d_{\mathcal{T}_f})$ (Def. \ref{t-f-ind}) is compact.
\end{lemma}

\begin{proof}
Since $\End(A)$ is a finite-dimensional vector space, it is sufficient to show that $\mathcal{T}_f$ is closed and bounded (Fact \ref{heine-borel}).
\begin{itemize}
\item Closedness: Let $(T_n)_n$ be a sequence in $\mathcal{T}_f$ that converges in $\End(A)$, i.e. $T_n \rightarrow T \in \End(A)$. Closedness of $\mathcal{T}_f$ can be shown by showing that the limit $T$ is an element of $\mathcal{T}_f$ (Fact \ref{closure-limits}), in other words by showing that $T$ is positive and that $u_A \circ T = f$.
\begin{itemize}
\item Positivity: Let $\omega \in A_+$. The map $T_n \mapsto T_n(\omega)$ is a linear map from $\End(A)$ to $A$. Thus, since $(T_n)_n \subseteq \mathcal{T}_f$ is convergent, the sequence $(T_n(\omega))_n \subseteq A_+$ is convergent as well, and the limit of $(T_n(\omega))_n$ coincides with $T(\omega)$ (Fact \ref{sequent-cont}). By the definition of an abstract state space (Def. \ref{ass-def}), $A_+$ is closed, so the limit $T(\omega)$ of $(T_n(\omega))_n$ is an element of $A_+$ (Fact \ref{closure-limits}), so $T$ is positive.
\item $u_A \circ T = f$: Note that the sequence $(u_A \circ T_n)_n$ in $A^*$ is constantly equal to $f$ and thus $u_A \circ T_n \rightarrow f$. On the other hand, the map $T_n \mapsto u_A \circ T_n$ is a linear map from $\End(A)$ to $A^*$, so $u_A \circ T_n \rightarrow u_A \circ T$ (Fact \ref{sequent-cont}) and thus $u_A \circ T = f$.
\end{itemize}
We have shown that $T$ is positive and that $u_A \circ T = f$. Therefore, $T_n \rightarrow T \in \mathcal{T}_f$, so $\mathcal{T}_f$ is closed.
\item Boundedness: Since any two norms on $\End(A)$ are equivalent (Fact \ref{norm-equivalence}), it is sufficient to show the boundedness of $\mathcal{T}_F$ for a particular choice of a norm on $\End(A)$. Choose the norm $\Vert T \Vert_{u_A} = \sup_{\omega \in \Omega_A} \vert u_A(T(\omega))\vert$. It is easily verified that this indeed gives a norm on $\End(A)$ (for positive definiteness, make use of the fact that $\spa(\Omega_A) = A$ since $A_+$ is generating). For all $T \in \mathcal{T}_f$, it holds that $u_A \circ T = f$ and thus $\Vert T \Vert_{u_A} = \sup_{\omega \in \Omega_A} \vert f(\omega) \vert \leq 1$, so $\mathcal{T}_f$ is bounded.
\end{itemize}
We have shown that $\mathcal{T}_f$ is closed and bounded, so by (Fact \ref{heine-borel}), $\mathcal{T}_f$ is compact.
\end{proof}

\begin{lemma}
\label{continuity-lemma}
For an abstract state space $A$ and an effect $f \in E_A$, it holds that for any norm $\Vert \cdot \Vert_A$ on $A$, the disturbance function $D_f$ (Def. \ref{disturb-f})
\begin{align*}
\left.\begin{array}{cccl}D_f: & \mathcal{T}_f & \rightarrow & \mathbb{R} \\ & T & \mapsto & \max\limits_{\omega \in F_f} \Vert T(\omega) - \omega \Vert_A\end{array}\right.
\end{align*}
is a continuous function on $\mathcal{T}_f$ with respect to the operator norm $\Vert \cdot \Vert_{\End(A)}$, c.f. (Def. \ref{t-f-ind}).
\end{lemma}

\begin{proof}
This is easily calculated. Let $I_A$ be the identity operator on $A$. For any $T, S \in \mathcal{T}_f$, we have that
\begin{align*}
\vert D_f(T) - D_f(S) \vert &= \left\vert \left( \max\limits_{\omega \in F_f} \Vert T(\omega) - \omega \Vert_A \right) - \left( \max\limits_{\sigma \in F_f} \Vert S(\sigma) - \sigma \Vert_A \right) \right\vert \\
&= \left\vert \left( \max\limits_{\omega \in F_f} \Vert (T - I_A)\omega \Vert_A \right) - \left( \max\limits_{\sigma \in F_f} \Vert (S - I_A)\sigma \Vert_A \right) \right\vert \\
&\leq \left\vert \max\limits_{\omega \in F_f} \bigg( \Vert (T - I_A)\omega \Vert_A - \Vert (S - I_A)\omega \Vert_A \bigg) \right\vert \\
&\leq \max\limits_{\omega \in F_f} \bigg\vert \Vert (T - I_A)\omega \Vert_A - \Vert (S - I_A)\omega \Vert_A \bigg\vert \\
&\leq \max\limits_{\omega \in F_f} \bigg\vert \Vert (T - I_A)\omega - (S - I_A)\omega \Vert_A \bigg\vert \\
&\leq \max\limits_{\omega \in F_f} \Vert (T - S)\omega \Vert_A \\
&\leq \underbrace{\left( \max\limits_{\omega \in F_f} \Vert \omega \Vert_A \right)}_{\text{const.}} \Vert T - S \Vert_{\End(A)} \,,
\end{align*}
so $D_f$ is continuous.
\end{proof}

\subsection{Proof of the theorem}
\label{approx-proof-section}

\begin{thm}
\label{approx-case-thm}
Let $A$ be a polytopic non-classical theory (Def. \ref{polytopic-theory}) and (Def. \ref{classical-theory}) and let $\Vert \cdot \Vert_A$ be any norm on $A$. Then there is a pure effect $f \in E_A$ and a positive number $\epsilon$ with the following property: For every transformation $T: A \rightarrow A$ with $f = u_A \circ T$, there is a state $\rho \in \Omega_A$ with $f(\rho) = 1$ and $\Vert T(\rho) - \rho \Vert_A \geq \epsilon$.
\end{thm}

\begin{proof}
By virtue of Lemma \ref{non-pres-f}, there is a pure effect $f \in E_A$ such that the certain face $F_f$ is a minus-face of $\Omega_A$ and such that
\begin{align}
\left\{ \begin{array}{l}
\text{there is no transformation } L: A \rightarrow A \\
\text{with } f = u_A \circ L \text{ and } L(\omega) = \omega \ \forall \omega \in F_f \,.
\end{array} \right. \label{l-no-trafo}
\end{align}
This is the effect $f$ for which we will show the existence of a number $\epsilon > 0$ with the claimed properties. We make a proof by cases, where we distinguish between the cases where the impossible face $\overline F_f$ of $f$ (Def. \ref{cert-face-def}) satisfies $\dim \overline F_f = 0$ and where $\dim \overline F_f \geq 1$ (the case $\dim \overline F_f = -1$ is not possible since $F_f$ is a minus-face of $\Omega_A$).

\begin{enumerate}[(i)]

\item Assume that $\dim \overline F_f = 0$, i.e.
\begin{align*}
\overline F_f = \{ \overline \omega_f \} \quad \text{for some } \overline \omega_f \in \Omega_A \,.
\end{align*}
Since $\overline \omega_f \notin \spa(F_f)$, it holds that $\spa(F_f) \cap \spa(\overline F_f) = \{0\}$. ($\overline \omega_f \notin \spa(F_f)$ can be verified using the fact that $\spa(F_f) = \aff(F_f \cup \{0\})$ (Fact \ref{spa-is-aff-0}), $F_f = \aff(F_f) \cap \Omega_A$ (Fact \ref{face-cap-aff}) and $u_A(\overline \omega_f) = 1$ but $u_A(0) = 0$.) Thus, there is a linear map $L: A \rightarrow A$ with
\begin{align}
&L(\omega) = \omega \quad \forall \omega \in F_f \,, \label{l-cond-1} \\
&L(\overline F_f) = \{0\} \,. \label{l-cond-2}
\end{align}
Note that
\begin{align}
\label{dim-f-f-a-2}
\dim F_f = \dim A - 2
\end{align}
(since $F_f \subseteq \Omega_A$ is a minus-face). Thus,
\begin{align*}
\dim(\spa(F_f)) + \dim(\spa(\overline F_f)) &= \dim F_f + 1 + \dim \overline F_f + 1 &\text{by (Fact \ref{0-notin-aff}) and (Fact \ref{dim-span-plus-one})}\\
&= (\dim A - 2) + 1 + 0 + 1 &\text{by (\ref{dim-f-f-a-2})}\\
&= \dim A \,,
\end{align*}
so the conditions (\ref{l-cond-1}) and (\ref{l-cond-2}) fully determine the linear map $L: A \rightarrow A$. It also means that (\ref{l-cond-1}) and (\ref{l-cond-2}) imply that $f = u_A \circ L$ on a set that spans $A$, and thus $f = u_A \circ L$ everywhere. Thus, by (\ref{l-no-trafo}), $L$ cannot be a transformation, so by (Def. \ref{trafo-def}), linearity or $L(\Omega_A) \subseteq \Omega_A^{\leq 1}$ must fail for $L$. However, we have constructed $L$ to be linear, so $L(\Omega_A) \subseteq \Omega_A^{\leq 1}$ must fail. Thus, there is a $\tau \in \Omega_A$ such that
\begin{align}
L(\tau) \notin \Omega_A^{\leq 1} \,, \quad \text{i.e.} \quad d(L(\tau), \Omega_A^{\leq 1}) > 0 \,, \label{tau-outside}
\end{align}
where $d( \, \cdot \, , \Omega_A^{\leq 1})$ is defined in (Def. \ref{set-distance-def}). Define
\begin{align*}
d_A := \dim A
\end{align*}and recall our definition of the dimension of a set, (Def. \ref{set-dim-def}). It holds that $\tau \in \Omega_A \subseteq \aff(F_f \cup \overline F_f)$. (This follows from the fact that $F_f$ is a minus-face of $\Omega_A$ and $F_f, \overline F_f \subset \Omega_A$ but $\overline F_f \notin \aff(F_f)$.) This implies that there must be $\dim F_f + 1 = (\dim A - 2) + 1 = (d_A - 1)$ points $\{ \omega_1, \ldots, \omega_{d_A-1} \} \subset F_f$ such that
\begin{align}
&\tau \in \aff( \{ \omega_1, \ldots, \omega_{d_A-1}, \overline \omega_f \} ) \,, \quad \text{i.e.} \nonumber \\
&\tau = \left( \sum\limits_{i=1}^{d_A-1} \alpha_i \omega_i \right) + \alpha_{d_A} \overline \omega_f \quad \text{for some coefficients } \{ \alpha_i \}_{i=1}^{d_A} \text{ with } \sum\limits_{i=1}^{d_A} \alpha_i = 1 \,. \label{tau-decomp}
\end{align}
We define
\begin{align*}
\alpha_\text{max} := \max \{ \vert \alpha_i \vert \mid i=1, \ldots, d_A-1 \} \,.
\end{align*}
Note that $\alpha_\text{max}$ is positive by (\ref{tau-decomp}) since $\tau \neq \overline \omega_f$ (we have chosen $\tau$ such that $L(\tau) \notin \Omega_A^{\leq 1}$, but $L(\widetilde \omega_f) = 0 \in \Omega_A^{\leq 1}$). Let $T: A \rightarrow A$ be any transformation with $f = u_A \circ T$. \emph{Assume} that
\begin{align}
\Vert T(\omega) - \omega \Vert_A < \frac{d(L(\tau), \Omega_A^{\leq 1})}{(d_A-1) \alpha_\text{max}} \quad \forall \omega \in F_f. \label{false-ass}
\end{align}
We will show that this leads to a contradiction to the assumption that $T$ is positive. This, in turn, will show that the term on the right hand side of Inequality (\ref{false-ass}) is the $\epsilon$ with the claimed property (note that the term is independent on the choice of the transformation $T$).

Together with Assumption (\ref{false-ass}), we can use the triangle-inequality for the norm to derive the following bound:
\begin{align}
\Vert T(\tau) - L(\tau) \Vert_A &= \left\Vert \left( \sum\limits_{i=1}^{d_A - 1} \alpha_i T(\omega_i) \right) + \alpha_{d_A} T(\overline \omega_f) - \left( \sum\limits_{j=1}^{d_A - 1} \alpha_j L(\omega_j) \right) - \alpha_{d_A} L(\overline \omega_f) \right\Vert_A &\text{by (\ref{tau-decomp})} \nonumber \\
&\leq \sum\limits_{i=1}^{d_A-1} \vert \alpha_i \vert \ \Vert T(\omega_i) - \underbrace{L(\omega_i)}_{\omega_i} \Vert_A + \vert \alpha_{d_A} \vert \ \Vert \underbrace{T(\overline \omega_f)}_{0} - \underbrace{L(\overline \omega_f)}_{0} \Vert_A \nonumber \\
&= \sum\limits_{i=1}^{d_A-1} \vert \alpha_i \vert \ \Vert T(\omega_i) - \omega_i \Vert_A \nonumber \\
&< (d_A-1) \alpha_\text{max} \frac{d(L(\tau), \Omega_A^{\leq 1})}{(d_A-1) \alpha_\text{max}} &\text{by (\ref{false-ass})} \nonumber
\end{align}
and therefore
\begin{align}
\Vert T(\tau) - L(\tau) \Vert_A < d(L(\tau), \Omega_A^{\leq 1}) \,. \label{t-l-ineq}
\end{align}
For any $\sigma \in \Omega_A^{\leq 1}$, we can use the triangle inequality again to derive the following inequality:
\begin{align}
\Vert L(\tau) - \sigma \Vert_A &\leq \Vert L(\tau) - T(\tau) \Vert_A + \Vert T(\tau) - \sigma \Vert_A \nonumber \\
\Rightarrow \quad \Vert T(\tau) - \sigma \Vert_A &\geq \Vert L(\tau) - \sigma \Vert_A - \Vert L(\tau) - T(\tau) \Vert_A \,. \label{useful-ineq}
\end{align}
This allows us to conclude
\begin{align*}
d(T(\tau), \Omega_A^{\leq 1}) &=\min\limits_{\tau \in \Omega_A^{\leq 1}} \Vert T(\tau) - \sigma \Vert_A \\
&\geq \min\limits_{\tau \in \Omega_A^{\leq 1}} \bigg( \Vert L(\tau) - \sigma \Vert_A - \Vert L(\tau) - T(\tau) \Vert_A \bigg) &\text{by (\ref{useful-ineq})} \\
&= \min\limits_{\tau \in \Omega_A^{\leq 1}} \bigg( \Vert L(\tau) - \sigma \Vert_A \bigg) - \Vert L(\tau) - T(\tau) \Vert_A \\
&> d(L(\tau), \Omega_A^{\leq 1}) - d(L(\tau), \Omega_A^{\leq 1}) \quad &\text{by (\ref{t-l-ineq})}
\end{align*}
and therefore
\begin{align*}
d(T(\tau), \Omega_A^{\leq 1}) > 0 \,.
\end{align*}
Thus, assumption (\ref{false-ass}) implies that there is a $\tau \in \Omega_A$ which is mapped outside of $\Omega_A^{\leq 1}$ by $T$, so it implies that the map $T$ is not positive. But $T$ is a transformation and therefore positive, so the assumption (\ref{false-ass}) must be wrong. The negation of (\ref{false-ass}) is
\begin{align}
\exists \rho \in F_f: \Vert T(\rho) - \rho \Vert_A \geq \frac{d(L(\tau), \Omega_A^{\leq 1})}{(d_A-1) \alpha_\text{max}} \,. \label{0-claim-proved}
\end{align}
The fact that $\rho \in F_f$ means that $f(\rho) = 1$. Set
\begin{align*}
\epsilon := \frac{d(L(\tau), \Omega_A^{\leq 1})}{(d_A-1) \alpha_\text{max}} \,.
\end{align*}
This is a positive number by (\ref{tau-outside}). Since $\epsilon$ is independent of $T$ and $T$ is an arbitrary transformation with $f = u_A \circ T$, (\ref{0-claim-proved}) means that we have proved the claim for the case where $\dim \overline F_f = 0$.

\item Assume that $\dim \overline F_f \geq 1$. Let $T: A \rightarrow A$ be a transformation such that $f = u_A \circ T$. According to Lemma \ref{t-f-dim-lemma}, we have
\begin{align*}
\dim(T(F_f)) &\leq \dim A - \dim \overline F_f - 2 \\
&\leq d_A - 3 \,, \quad \text{where} \quad d_A := \dim A \,.
\end{align*}
Note that since $\Omega_A$ is a non-classical polytope, it holds that $\dim A \geq 3$ (for $\dim A = 1, 2$, the set $\Omega_A$ is a point or a line, respectively, which both are a simplex and therefore classical). The inequality shows that the minus-face $F_f$ of $\Omega_A$, which is a polytope with dimension $\dim F_f = d_A - 2$, is mapped to a set $T(F_f)$ which is at most $(d_A-3)$-dimensional. Diagrammatically,
\begin{align*}
\underbrace{F_f}_{(d_A-2)\text{-dim.}} \mapsto \underbrace{T(F_f)}_{\leq (d_A-3)\text{-dim.}} \subset \Omega_A \,.
\end{align*}
Therefore, $F_f$ cannot be contained in $T(F_f)$. In particular, there must be an $\omega \in F_f$ such that $\Vert T(\omega) - \omega \Vert_A > 0$. Since the map $\omega \mapsto \Vert T(\omega) - \omega \Vert _A$ is a continuous function and the polytope $F_f$ is compact (Fact \ref{polytope-compact}), the map attains a maximum on $F_f$. So far, we have shown the following: For every transformation $T: A \rightarrow A$ with $f = u_A \circ T$, it holds that $\max_{\omega \in F_f} \Vert T(\omega) - \omega \Vert_A$ is positive. Recapitulate our previous definitions (Def. \ref{t-f-ind}), (Def. \ref{disturb-f}). The transformations under consideration are given by
\begin{align*}
\mathcal{T}_f = \{ T \in \End(A) \mid T \text{ is positive, } u_A \circ T = f \} \,,
\end{align*}
This definition allows us to write the disturbance function $D_f$ as a function on $\mathcal{T}_f$:
\begin{align*}
\left.\begin{array}{cccl}D_f: & \mathcal{T}_f & \rightarrow & \mathbb{R} \\ & T & \mapsto & \max\limits_{\omega \in F_f} \Vert T(\omega) - \omega \Vert_A\end{array}\right.
\end{align*}
Using these definitions, we can summarize what we have proved so far by stating that the disturbance function $D_f$ is a positive function on $\mathcal{T}_f$. From Lemmas \ref{compactness-lemma} and \ref{continuity-lemma}, we know that $\mathcal{T}_f$ is compact and that $D_f$ is continuous on $\mathcal{T}_f$. A continuous function on a compact space attains a minimum, so $\min_{T \in \mathcal{T}_f} D_f(T)$ exists. Since the minimum of a positive function must be positive, it holds that
\begin{align*}
\min\limits_{T \in \mathcal{T}_f} D_f(T) = \min\limits_{T \in \mathcal{T}_f} \left( \max\limits_{\omega \in F_f} \Vert T(\omega) - \omega \Vert_A \right) > 0 \,.
\end{align*}
Set
\begin{align*}
\epsilon := \min\limits_{T \in \mathcal{T}_f} \left( \max\limits_{\omega \in F_f} \Vert T(\omega) - \omega \Vert_A \right) \,.
\end{align*}
Thus, $\epsilon$ is a positive number such that for every $T \in \mathcal{T}_f$, there is a $\rho \in F_f$ such that
\begin{align*}
\Vert T(\rho) - \rho \Vert_A \geq \epsilon \,.
\end{align*}
Writing out the definitions of $\mathcal{T}_f$ and $F_f$, (Def. \ref{t-f-ind}) and (Def. \ref{cert-face-def}), we have proved the existence of an $\epsilon > 0$ with the property that for every transformation $T: A \rightarrow A$ with $u_A \circ T = f$, there is a $\rho \in \Omega_A$ with $f(\rho) = 1$ such that $\Vert T(\rho) - \rho \Vert_A \geq \epsilon$, so we have proved the claim. \qedhere

\end{enumerate}
\end{proof}

\end{appendix}
\end{widetext}

\end{document}